\documentclass[10pt,journal,compsoc]{IEEEtran}
\usepackage{color}
\usepackage{amsmath}
\usepackage{amssymb}
\usepackage{dsfont}
\usepackage{upgreek}
\usepackage{booktabs}
\usepackage{multirow}
\usepackage[table,xcdraw]{xcolor}
\usepackage{subfigure}
\usepackage{bm}
\usepackage{xspace}
\usepackage{hyperref}

\usepackage{ifpdf}
 \ifpdf
 \else
 \fi
\usepackage{cite}
\ifCLASSINFOpdf
  \usepackage[pdftex]{graphicx}
  \graphicspath{{../pdf/}{../jpeg/}}
  \DeclareGraphicsExtensions{.pdf,.jpeg,.png}
\else
  \usepackage[dvips]{graphicx}
  \graphicspath{{../eps/}}
  \DeclareGraphicsExtensions{.eps}
\fi
\usepackage{caption2}  %Image forced Center

\usepackage{amsmath}
\usepackage{ntheorem}
\theoremstyle{plain}
\newtheorem{theorem}{Theorem}
\newtheorem{lemma}{Lemma}
\newtheorem*{proof}{Proof}
\theoremstyle{definition}
\newtheorem{definition}{Definition}
\newtheorem{example}{Example}
\usepackage[numbers,sort&compress]{natbib}
\usepackage{enumerate}
\usepackage{algorithm}
\usepackage{algorithmic}
\usepackage{array}

\usepackage{url}
\usepackage{multirow}
\usepackage{bm}
\begin{document}
\title{Order-preserving pattern mining with forgetting mechanism}

\author{Yan Li, Chenyu Ma, Rong Gao, Youxi Wu, \IEEEmembership{Senior Member, IEEE}, Jinyan Li, Wenjian Wang, and Xindong Wu, \IEEEmembership{Fellow,~IEEE}
\thanks{Manuscript received December 8, 2023.  (Corresponding authors: Youxi Wu and Jinyan Li). }
\thanks{Yan Li, Chenyu Ma, and Rong Gao are with the School of Economics and Management, Hebei University of Technology, Tianjin, 300400, China (e-mail: lywuc@163.com; mcy\_machenyu@163.com; rgao@hebut.edu.cn).}

\thanks{Youxi Wu is with the School of Artificial Intelligence, Hebei University of Technology, Tianjin, 300400, China (e-mail: wuc567@163.com).}

%and Computer Science and Control Engineering, Shenzhen Institute of Advanced Technology, Chinese Academy of Sciences, Shenzhen, 518055, China

\thanks{Jinyan Li is with the School of Computer Science and Control Engineering, Shenzhen University of Advanced Technology, Shenzhen 518055, China, 
and Shenzhen Institute of Advanced Technology, Chinese Academy of Sciences, Shenzhen 518055, China  (e-mail: Jinyan.li@siat.ac.cn).}

\thanks{Wenjian Wang is with the School of Computer and Information Technology, Shanxi University, Taiyuan, 237016, China (e-mail: wjwang@sxu.edu.cn).}

\thanks{Xindong Wu is with the Key Laboratory of Knowledge Engineering with Big Data (the Ministry of Education of China), Hefei University of Technology, Hefei 230009, China (e-mail: xwu@hfut.edu.cn).}%with Key Laboratory of Knowledge Engineering with Big Data (the Ministry of Education of China), Hefei University of Technology, Hefei, 230009, China (e-mail: xwu@hfut.edu.cn)}
}


\markboth{IEEE Transactions on Knowledge and Data Engineering }
{Shell \MakeLowercase{\textit{Ma et al.}}: Bare Demo of IEEEtran.cls for IEEE Journals}





\IEEEtitleabstractindextext{
\begin{abstract}
Order-preserving pattern (OPP) mining is a type of sequential pattern mining method in which a group of ranks of time series is used to represent an OPP. This approach can discover frequent trends in time series. Existing OPP mining algorithms consider data points at different time to be equally important; however, newer data usually have a more significant impact, while older data have a weaker impact. We therefore introduce the forgetting mechanism into OPP mining to reduce the importance of older data. This paper explores the mining of OPPs with forgetting mechanism (OPF) and proposes an algorithm called OPF-Miner that can discover frequent OPFs. OPF-Miner performs two tasks, candidate pattern generation and support calculation. In candidate pattern generation, OPF-Miner employs a maximal support priority strategy and a group pattern fusion strategy to avoid redundant pattern fusions. For support calculation, we propose an algorithm called support calculation with forgetting mechanism, which uses prefix and suffix pattern pruning strategies to avoid redundant support calculations. The experiments are conducted on nine datasets and 12 alternative algorithms. The results verify that OPF-Miner is superior to other competitive algorithms. More importantly, OPF-Miner yields good clustering performance for time series, since the forgetting mechanism is employed.

\end{abstract}

\begin{IEEEkeywords}
pattern mining, time series, forgetting mechanism, order-preserving pattern
\end{IEEEkeywords}
}
\maketitle

\IEEEdisplaynontitleabstractindextext
\IEEEpeerreviewmaketitle

\IEEEpeerreviewmaketitle

\section{Introduction}

%DNA testing, stock prediction, bio-genetics
%remaining useful life prediction \cite {sensors},
 
\IEEEPARstart{A}{} time series is an observation sequence measured over a continuous period of time \cite {tsdatamining, multivariate,jessica2020}, and is important in many fields of application, such as the analysis of stock data \cite {6tkdd2022NTP,stock},  environment forecasting \cite {environmentforecasting,environment1},  and traffic-flow prediction \cite {8tkdd2022neg}. Many methods have been employed for the analysis of time series data, such as generative adversarial networks \cite {GAN1,GAN2}, long short-term memory neural networks \cite {LSTM}, matrix profiles \cite {matrixprofile}, and sequential pattern mining \cite {nosep}.

Sequential pattern mining is an important data mining method \cite {20tkdd2019gan,lijinyan2017, ernsp} that can be used to discover interesting patterns in time series with good interpretability \cite {fnsp, bigdata}, and is attracting attention from researchers \cite {window}. Previous studies have been based on the idea that a time series needs to be discretized into symbols before sequential pattern mining methods can be used for time series mining \cite {31pyramidpattern,61li2022apind}. However, existing discretization methods, such as piecewise aggregate approximation \cite {Lin2002Finding} and symbolic aggregate approximation \cite {Keogh2005HOT}, have an excessive focus on the values of the time series data, making it difficult to use existing sequential pattern mining methods to mine frequent trends from such data.

Inspired by order-preserving pattern (OPP) matching \cite {kimopp,orderpreserving2} in the field of theoretical computer science, an OPP mining approach  \cite {OPPminer} was proposed in which the relative order of real values was used to represent a pattern that could be used to characterize the trend of time series, and an algorithm called OPP-Miner was developed to mine frequent trends from time series. To further improve the efficiency of OPP-Miner, an efficient OPP mining algorithm EFO-Miner \cite {OPRminer} was developed to mine OPPs. Moreover, an order-preserving rule (OPR) mining method was devised to discover the implicit relationships between OPPs \cite {OPRminer}. 

There are two main issues with existing OPP mining methods:

1) They adopt pattern fusion methods to generate candidate patterns, which involve some redundant pattern fusions. Hence, the efficiency of OPP-Miner  \cite {OPPminer} and EFO-Miner  \cite {OPRminer} can be improved. 

2) More importantly, they consider data points at different time to have equal importance. However, newer data typically have a significant impact, while the impact of older data is relatively weak. Hence, in OPP mining, it is necessary to assign different levels of importance to data at different time, which is similar to the forgetting mechanism of the human brain \cite {forgetting1}.  A forgetting curve is shown in Fig. \ref{Example of a forgetting curve}.
\begin{figure}
    \centering
    \includegraphics[width=0.95\linewidth]{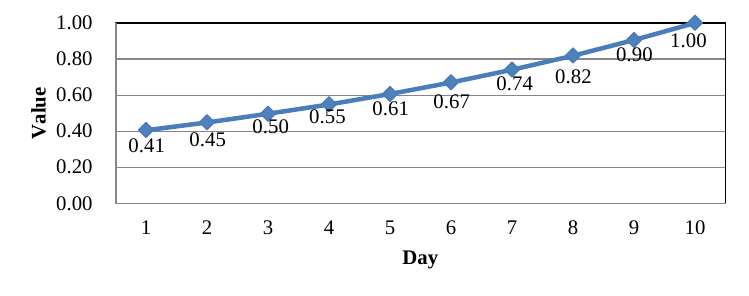}
    \caption{Example of a forgetting curve}
    \label{Example of a forgetting curve}
\end{figure}
%: over a period of 10 days, the human brain has the strongest memory on the 10th day, with previous memories gradually decreasing.

Inspired by the forgetting mechanism \cite {forgetting2}, we introduce it into sequential pattern mining to reduce the importance of old data. In this study, we explore the mining of OPPs with the forgetting mechanism (OPFs) and propose an effective algorithm named OPF-Miner to discover frequent OPFs. The main contributions of the paper are as follows.
\begin{enumerate}[1.]

\item To avoid the excessive influence of older data in a time series, we propose an OPF mining scheme that assigns different levels of importance to data at different time in the series, and develop an algorithm called OPF-Miner to mine OPFs effectively.

\item  To generate candidate patterns, we propose a maximal support priority strategy to improve the efficiency in a heuristic manner and a group pattern fusion (GP-Fusion) strategy to reduce redundant pattern fusions.

\item  To calculate the support, we propose an algorithm called support calculation with forgetting mechanism (SCF) in which prefix and suffix pattern pruning strategies are applied to further reduce redundant support calculations. 

\item  Our experimental results verify that OPF-Miner has better performance than other competitive algorithms. More importantly, clustering experiments validate that OPF-Miner can be used to realize feature extraction with good performance.
\end{enumerate}

The rest of this paper is organized as follows. Section \ref{section:Related work}  describes related work. Section \ref{section:Problem definitions}  presents a definition of the problem. Section \ref{section:Proposed algorithm} introduces the OPF-Miner algorithm, and in Section \ref{section:Experimental Results and Analysis}, its performance is validated. Section \ref{section:CONCLUSION} concludes this paper.

\section{Related work}
\label{section:Related work}

Sequential pattern mining (SPM) is a very important data mining method \cite {patternmining}, since it is easy to understand and can yield results with good interpretability. SPM is typically used to discover subsequences (called patterns) with high frequency, but there are also certain special forms, such as high-utility pattern mining for the discovery of low-frequency but important patterns \cite {unilyunutility, 10tkdd2022utility, tmis2022, 47RUP-Miner}, rare pattern mining for low-frequency patterns \cite {rare2016tkdd}, co-occurrence pattern mining for mining patterns with common prefixes \cite {paacoguo,mcor}, outlier pattern mining for abnormal cases \cite {38tkdd2020}, and two-way and three-way SPM for mining patterns with different levels of interests \cite {OWSP2021tmis, 5ins2020tri}. To reduce the number of patterns mined in this process, approaches such as closed SPM \cite {closed1,closed2, 39closed2020kbs}, maximal SPM \cite {maximalspm,max2022apin}, and top-$k$ SPM \cite {cooccpattern,2022constrast} also have been developed. Various mining methods have been designed to meet the needs of users, such as negative sequential pattern mining for mining patterns with missing items \cite {18_Wang2021,negativespm, nonoccurring}, SPM with gap constraints (repetitive SPM) to discover repetitively occurring patterns \cite {tkdd2012, 34biological,tongoverlapping, RNPminer}, episode mining with timestamps \cite{episodepattern, episoderule}, process pattern mining for event logs \cite {liucong1,liucong2}, incremental sequential pattern mining for uncertain databases \cite {incrementalmining}, and OPP mining for time series \cite {OPPminer, OPRminer}.

To analyze time series data, traditional methods first employ a symbolization method to discretize the original real values into symbols, and then apply an SPM method to mine the patterns. However, traditional methods focus too much on the values of time series data and overlooks the trend information of the time series. To tackle this issue, OPP mining \cite {OPPminer} based on OPP matching \cite {fastorder,filtrationorder} was proposed, in which a group of relative orders is used to represent a pattern and the number of occurrences of a pattern in a time series is calculated. OPP-Miner was the first algorithm for OPP mining to employ a pattern matching method and filtration and verification strategies to calculate the support. To improve the efficiency of support calculation, EFO-Miner \cite {OPRminer} utilized the results for the sub-patterns to calculate the occurrences of super-patterns. To further reveal the implicit relationships between OPPs, OPR mining was proposed and an algorithm called OPR-Miner was designed \cite {OPRminer}. OPP-Miner \cite {OPPminer} and OPR-Miner \cite {OPRminer} can discover frequent OPPs and strong OPRs, respectively. However, it is difficult to reflect the differences between two classes of time series using these methods. To overcome this drawback, COPP-Miner was proposed \cite {COPPminer}, which can extract contrast OPPs as features for time series classification. 

Existing OPP mining algorithms consider data points at different time to be equally important. However, newer data usually have a more significant impact, while older data have a weaker impact. To address this issue, we introduce the forgetting mechanism into OPP mining to reduce the importance of old data, and propose an OPF mining scheme to mine OPPs with the forgetting mechanism.

\section{Problem definition}
\label{section:Problem definitions}

\begin{definition} \label{definition1}
  \rm (Time series \cite{Lin2002Finding}) A time series consists of numerical values of the same statistical indicator arranged in order, and is denoted as $\textbf{t}=(t_1,\cdots,t_i,\cdots,t_n)$, where $1\leq i\leq n$.
\end{definition}

\begin{definition}\label{definition2}
   \rm (Rank and relative order \cite{kimopp}) Given a time series $\textbf{t}=(t_1,\cdots,t_i,\cdots,t_n)$, the rank of element $t_i$, denoted by $rank_\textbf{t}(t_i)$, is defined by $1+y$, where $y$ is the number of elements less than $t_i$ in time series \textbf{t}. The relative order of $\textbf{t}$, denoted by $R(\textbf{t})$, is defined by $R(\textbf{t})=(rank_\textbf{t}(t_1), rank_\textbf{t}(t_2),\cdots, rank_\textbf{t}(t_n))$.
\end{definition}

\begin{definition}\label{definition3}
  \rm  (Order-preserving pattern, OPP \cite{kimopp}) 
A pattern represented by the relative order of a sub-time series is called an OPP.
 \end{definition}

 \begin{example}\label{example1}
Given a pattern $\textbf{p}=(15,32,29)$, we know that 15 is the smallest value in \textbf{p}, i.e., $rank_\textbf{p}(15)=1$. Similarly, $rank_\textbf{p}(32)=3$ and $rank_\textbf{p}(29)=2$. Thus, the OPP of \textbf{p} is $R(\textbf{p})=(1,3,2)$.
\end{example}

\begin{definition}\label{definition4}
\rm  (Occurrence \cite{OPPminer}) Given a time series $\textbf{t}=(t_1,\cdots,t_i,\cdots,t_n)$ and a pattern $\textbf{p}=(p_1, p_2,\cdots,p_m)$, if there exists a sub-time series 
   $\textbf{t}^{'}=(t_i, t_{i+1},..., t_{i+m-1})$ ($1 \leq i$ and $i+m-1\leq n$) satisfying $R(\textbf{t}^{'})=R(\textbf{p})$, then $\textbf{t}^{'}$ is an occurrence of pattern \textbf{p} in \textbf{t}, and the last position of $\textbf{t}^{'}$ $<$$i+m-1$$>$ is used to denote an occurrence.

\end{definition}

\begin{definition}\label{definition5}
  \rm  (Support with forgetting mechanism) The support of pattern \textbf{p} with the forgetting mechanism (or simply the support) is the sum of the forgetting values at the position of each occurrence, denoted by $fsup(\textbf{p},\textbf{t})$, which can be calculated using Equation \ref{equation1}: 
    \begin{equation}
     fsup(\textbf{p},\textbf{t})=\sum_{j=1}^{n} f_j
     \label{equation1}
    \end{equation}
where $f_j=e^{(-k\times(n-j))}$ is the forgetting value of position $<$$j$$>$, $k(0<k<1)$ is the forgetting factor, which can be set to $1/n$, $n$ is the length of the time series, and $<$$j$$>$ $(1\leq j\leq n)$ is the position at which pattern \textbf{p} occurs.  
\end{definition}

\begin{definition}\label{definition6}
 \rm(Frequent order-preserving pattern with forgetting mechanism, OPF) If the support of pattern \textbf{p}  with forgetting mechanism in a time series \textbf{t}  is no less than the minimum support threshold $minsup $, then pattern \textbf{p} is called a frequent OPF.
\end{definition}

\textbf{OPF mining:} The aim of OPF mining is to discover all frequent OPFs in a time series.

\begin{example}\label{example2}
Suppose we have a time series $\textbf{t}=(15,32,29,27,34,33,25,20,28,23)$ with corresponding indices as shown in Table \ref{Element indexes for time series} and a pattern $\textbf{p}=(1,3,2)$. The forgetting factor $k=1/10=0.1$, since the length of \textbf{t} is 10.

\begin{table} [h]
\centering
\tiny
    %% \vspace{-1.2em}
    \scriptsize
       \caption{Element indexes for time series \textbf{t}}
    \begin{tabular}{lcccccccccc}
        \toprule
        %\textbf{Type}
         \textbf{ID} &\textbf{1} &\textbf{2}&\textbf{3}&\textbf{4}&\textbf{5}&\textbf{6}&\textbf{7}&\textbf{8}&\textbf{9}&\textbf{10}\\
        \midrule
        \textbf{t}& 15	&32&	29&	27&	34&	33&	25 &	20&	28&	23\\
         
        \bottomrule
    \end{tabular}
	\label{Element indexes for time series}
    % \vspace{-0.3cm}
\end{table}

From Example \ref{example1} , we know that $(t_1,t_2,t_3)$ is an occurrence of pattern \textbf{p} in time series \textbf{t}, denoted as $<$3$>$, since $R(t_1,t_2,t_3)=(1,3,2)$. Similarly, $(t_4,t_5,t_6)$ and $(t_8,t_9,t_{10})$ are another two occurrences of pattern \textbf{p} in time series \textbf{t}. There are then three occurrences of pattern \textbf{p} in time series \textbf{t}: $\{<$3$>$,$<$6$>$,$<$10$>$$\}$. The forgetting factor $k=0.1$. Hence, $fsup(p,t)=\sum_{j=1}^{n}f_j=f_3+f_6+f_{10}=e^{(-0.1\times(10-3))}+ e^{(-0.1\times(10-6))}+e^{(-0.1\times(10-10))}=2.17$. If $minsup=1.5$, then \textbf{p} is a frequent OPF. Similarly, all frequent OPFs in \textbf{t} are $F=\{(1,2), (2,1), (3,2,1), (2,1,3), (1,3,2)\}$.
\end{example}

The symbols used in this paper are shown in Table \ref{Notation}.

\begin{table}
\centering
    % \vspace{-1.2em}
    %\footnotesize{
    \scriptsize
    \caption{Notations}
       \begin{tabular}{lc}
        \toprule
        {Symbol}     & {Description}  \\
    %      &    &    & {constraint}          & {confidence}   \\
        \midrule
        {\textbf{t}}  & A time series $(t_1,\cdots,t_i,\cdots,t_n)$ with length $n$   \\
            {$rank_\textbf{t}(t_i)$}  & The rank of element $t_i$  \\
          {$R(\textbf{t})$} & The relative order of $\textbf{t}$\\
          {\textbf{p}}  & A pattern $(p_1,p_2,\cdots,p_m)$ with length $m$   \\
        {$fsup(\textbf{p},\textbf{t})$} &The support of pattern \textbf{p} with forgetting mechanism \\
         {$k$}  & The  forgetting factor\\
         {$minsup$}  & The minimum support threshold    \\
          {$L_\textbf{p}$}  & The occurrences of pattern \textbf{p}  \\
         {$pre_\textbf{p}$}  & The prefix set of \textbf{p}  \\
         {$suf_\textbf{p}$}  & The suffix set of \textbf{p}  \\
        {$\textbf{f}$}  & The forgetting values of time series \textbf{t}  \\
         {$f_j$}  & The forgetting value for  position $<$$j$$>$  \\
        \bottomrule
    \end{tabular}
	\label{Notation}
    % \vspace{-0.3cm}
\end{table}

\section{Proposed algorithm}
\label{section:Proposed algorithm}
To tackle the task of OPF mining, we propose OPF-Miner, which consists of two parts, the generation of candidate patterns, and the calculation of support. To effectively reduce the number of candidate patterns, we propose GP-Fusion and maximal support priority strategies to improve the efficiency, as described in Section \ref{subsection:Generation of candidate patterns}. To avoid the excessive influence of older data in the time series, we propose an algorithm called SCF which applies two pruning strategies, as explained in Section \ref{Support calculation with forgetting mechanism}. Section \ref{OPF-Miner} describes the OPF-Miner algorithm for discovering all frequent OPFs. Finally, Section \ref{complexity} analyzes the time and space complexity of OPF-Miner. The framework of OPF-Miner is shown in Fig. \ref{Framework of OPF-Miner}.

\begin{figure}
    \centering
    \includegraphics[width=1\linewidth]{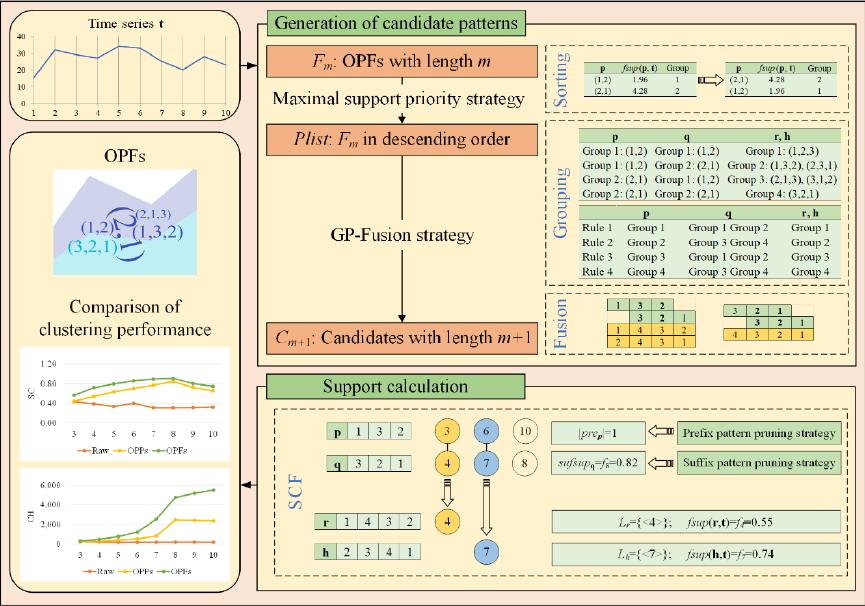}
    \caption{Framework of OPF-Miner. OPF-Miner contains two parts: generation of candidate patterns and support calculation. In the generation of candidate patterns, OPF-Miner adopts a maximal support priority strategy and a GP-Fusion strategy to avoid redundant pattern fusions. In the support calculation, an algorithm called SCF is employed, which uses prefix and suffix pattern pruning strategies to avoid redundant support calculations.}

    \label{Framework of OPF-Miner}
\end{figure}

\subsection{Generation of candidate patterns}
\label{subsection:Generation of candidate patterns}

In this section, we first review the pattern fusion method, and then propose a strategy called GP-fusion to reduce the calculations involved in pattern fusion. Finally, a maximal support priority strategy is proposed to further improve the efficiency.

\subsubsection{Pattern fusion}
To generate candidate patterns, traditional methods use an enumeration strategy, which can generate many redundant patterns. To overcome this problem, a pattern fusion method was proposed {\cite{OPPminer}}. Related definitions and illustrative examples are given .

\begin{definition}\label{definition7}
   \rm (Prefix OPP and suffix OPP) Given an OPP $\textbf{p}=(p_1,p_2,\cdots,p_m)$, \textbf{e}=\textit{prefixop}(\textbf{p})=$R(p_1,p_2,\cdots,p_{m-1})$ is called the prefix OPP of \textbf{p}, and \textbf{k}=\textit{suffixop}(\textbf{p})=$R(p_2,\cdots,p_m)$ is called the suffix OPP of \textbf{p}. Thus, patterns \textbf{e} and \textbf{k} are the order-preserving sub-patterns of \textbf{p}, and \textbf{p} is the order-preserving super-pattern of \textbf{e} and \textbf{k}.
\end{definition}

\begin{definition}\label{definition8}
   \rm (Pattern fusion) For $\textbf{p}=(p_1,p_2,\cdots,p_m)$ and $\textbf{q}=(q_1,q_2,\cdots,q_m)$, if \textit{suffixop}(\textbf{p})=\textit{prefixop}(\textbf{q}), then \textbf{p} and \textbf{q} can generate a super-pattern with length $m+1$. There are two possible cases:
   
\rm Case 1: If $p_1=q_m$, then \textbf{p} and  \textbf{q} 
can be fused into two patterns  $\textbf{r}=(r_1,r_2,\cdots,r_{m+1})$ and 
$\textbf{h}=(h_1,h_2,\cdots,h_{m+1})$, denoted as 
$\textbf{r}, \textbf{h}=\textbf{p}\oplus\textbf{q}$. $r_i$ and $h_i$ can be calculated as follows.

For pattern \textbf{r}, $r_1=p_1$ and $r_{m+1}=q_{m}+1$. If $q_i<p_1$, then $r_{i+1}=q_i$. Otherwise, $r_{i+1}=q_{i}+1$ $(1\leq i\leq m-1)$.

For pattern \textbf{h}, $h_1=p_{1}+1$ and $h_{m+1}=q_m$. If $p_i<q_m$, then $h_i=p_i$. Otherwise, $h_i=p_{i}+1$ $(2\leq i\leq m)$.

Case 2: If $p_1\neq q_m$, then \textbf{p} and \textbf{q} can be fused into one pattern $r=(r_1,r_2,\cdots,r_{m+1})$, denoted as $ \textbf{r}=\textbf{p}\oplus \textbf{q}$. $r_i$ can be calculated as follows.

If $p_1<q_m$, then $r_1=p_1$ and $r_{m+1}=q_m+1$. If $q_i<p_1$, then $r_{i+1}=q_i$. Otherwise, $r_{i+1}=q_i+1$ $(1\leq i\leq m-1)$.

If $p_1>q_m$, then $r_1=p_1+1$ and $r_{m+1}=q_m$. If $p_i<q_m$, then $r_i=p_i$. Otherwise, $r_i=p_i+1$ $(2\leq i\leq m)$.
\end{definition}

\begin{lemma}
The pattern fusion strategy is correct and complete \cite{OPRminer}, i.e., all feasible candidate patterns are generated and redundant candidate patterns are filtered. 
\end{lemma}\label {lemma1}

Example \ref{example3} illustrates the method of generating candidate patterns.

\begin{example}\label{example3} 
Given two patterns \textbf{p}=(1,3,2) and \textbf{q}=(3,2,1), there are two possible cases for the generation of candidate patterns with length four.

We know that \textit{prefixop}(\textbf{p})=R(1,3)=(1,2), since 1 is the smallest in (1,3), and 3 is the second smallest in (1,3). Similarly, \textit{suffixop}(\textbf{p})=R(3,2)=(2,1), \textit{prefixop}(\textbf{q})=R(3,2)=(2,1), and suffixop(\textbf{q})=R(2,1)=(2,1). Therefore, \textit{suffixop}(\textbf{p})=\textit{prefixop}(\textbf{q})=(2,1) and \textit{suffixop}(\textbf{q})= \textit{prefixop}(\textbf{q})=(2,1). Hence, \textbf{p} can fuse with \textbf{q}, and \textbf{q} can also fuse with \textbf{q}.

The fusion of \textbf{p}  and \textbf{q} satisfies Case 1, since \textit{suffixop}(\textbf{p})=\textit{prefixop}(\textbf{q})=(2,1), and $p_1=q_3=1$. 
Thus, we can generate two candidate patterns $\textbf{r}, \textbf{h}=\textbf{p}\oplus\textbf{q}$. For pattern \textbf{r}, $r_1=p_1=1$ and $r_4=q_3+1=2$. Moreover, $(r_2,r_3)=(q_1+1,q_2+1)=(4,3)$, since $q_1$$>$$p_1$ and $q_2$$>$$p_1$. Thus, pattern $\textbf{r}=(1,4,3,2)$. Similarly, pattern $\textbf{h}=(2,4,3,1)$.

The fusion of \textbf{q} and \textbf{q} satisfies Case 2, since \textit{suffixop}(\textbf{q})=\textit{prefixop}(\textbf{q})=(2,1), and $q_1$$>$$q_3$. 
Thus, we can generate one candidate pattern $\textbf{r}=\textbf{q}\oplus\textbf{q}$. For pattern \textbf{r}, $r_1=q_1+1=4$ and $r_4=q_3=1$. Moreover, $(r_2,r_3)=(q_2+1,q_3+1) =(3,2)$,   since $q_2$$>$$q_3$ and $q_3$$=$$q_3$. Thus, pattern $\textbf{r}=(4,3,2,1)$.

Therefore, the pattern fusion strategy generates three candidate patterns. However, the enumeration strategy generates four candidate patterns based on a pattern and thus generates eight candidate patterns based on \textbf{p} and \textbf{q}. Hence, the pattern fusion strategy outperforms the enumeration strategy, since the pattern fusion strategy generates fewer candidate patterns, and the fewer candidate patterns, the better the performance.

\end{example}

\subsubsection{GP-Fusion}

Although pattern fusion improves the mining efficiency, it has certain drawbacks. Suppose there are $l$ patterns with length $m$ stored in set $F_m$. The pattern fusion algorithm selects a pattern \textbf{p} from set $F_m$, and then has to check each pattern in set $F_m$. Thus, for pattern \textbf{p}, the pattern fusion method has to check $F_m$ a total of $l$ times, meaning that $l^2$ checks are carried out to generate all candidate patterns. To improve the efficiency, we propose a method called GP-Fusion which involves the following steps.

Step 1: Store patterns with length two in two groups. Store patterns (1,2) and (2,1) in Groups 1 and 2, respectively.

Step 2: Fuse patterns with length two into four groups. 
$\textbf{r}, \textbf{h}=\textbf{p}\oplus\textbf{q}$ are stored in four groups and the fusion results are shown in Table \ref{Fusion results for patterns with length two}.

\begin{table} [h]
\centering
    % \vspace{-1.2em}
    \scriptsize
       \caption{Fusion results for patterns with length two}
    \begin{tabular}{lcc}
        \toprule
        %\textbf{Type}
         \textbf{p} &\textbf{q} &\textbf{r},\textbf{h}\\
        \midrule
        Group 1: (1,2)& Group 1: (1,2)	&Group 1: (1,2,3)\\
        Group 1: (1,2)& Group 2: (2,1)	&Group 2: (1,3,2), (2,3,1)\\
         Group 2: (2,1)& Group 1: (1,2)	&Group 3: (2,1,3), (3,1,2)\\
          Group 2: (2,1)& Group 2: (2,1)	&Group 4: (3,2,1)\\
        \bottomrule
    \end{tabular}
	\label{Fusion results for patterns with length two}
    % \vspace{-0.3cm}
\end{table}

Step 3: Fuse patterns with length $m$ ($m$$>$2) into four groups.  The fusion rules are shown in Table \ref{Fusion rules for patterns}. We take Rule 1 as an example, and note that GP-Fusion selects a pattern \textbf{p} from Group 1 and then a pattern \textbf{q} from Groups 1 and 2. Finally, super-patterns \textbf{r} and \textbf{h} are stored in Group 1.

\begin{table} [h]
\centering
    % \vspace{-1.2em}
    \scriptsize
       \caption{Fusion rules for patterns with length $m$ ($m$$>$2)}
    \begin{tabular}{lccc}
        \toprule
        %\textbf{Type}
         \textbf{}&\textbf{p} &\textbf{q} &\textbf{r},\textbf{h}\\
        \midrule
        Rule 1&Group 1& Group 1	Group 2&Group 1\\
        Rule 2&Group 2& Group 3	Group 4&Group 2\\
        Rule 3&Group 3& Group 1	Group 2&Group 3\\
        Rule 4&Group 4& Group 3	Group 4&Group 4\\
        \bottomrule
    \end{tabular}
	\label{Fusion rules for patterns}
    % \vspace{-0.3cm}
\end{table}

Before proving that the GP-Fusion strategy is complete, we first prove two theorems to clarify the characteristics of the patterns in the four groups.

\begin{theorem}\label {theorem1}
 The first two ranks of each pattern in Groups 1 and 2 are (1,2), and those for each pattern in Groups 3 and 4 are (2,1). 
\end{theorem}
\begin{proof}
When the length of the pattern is three, the patterns in Groups 1 and 2 are \{(1,2,3)\} and \{(1,3,2), (2,3,1)\}, respectively. Thus, the first two ranks of each pattern with length three in Groups 1 and 2 are (1,2). Now, we suppose the first two ranks of each pattern with length $m$ in Groups 1 and 2 are (1,2). We will prove that the first two ranks of each pattern of length $m+1$ in Groups 1 and 2 are also (1,2). According to Table \ref{Fusion rules for patterns}, each pattern \textbf{r} in Groups 1 and 2 is generated by $\textbf{p}\oplus\textbf{q}$, 
where $\textbf{p}=(p_1,p_2,\cdots,p_m)$ is in Groups 1 and 2. Since the first two ranks of each pattern \textbf{p} are (1,2), i.e., $R(p_1,p_2)=(1,2)$, the first two ranks of each pattern \textbf{r} are also (1,2). Hence, the first two ranks of each pattern with length $m+1$ in Groups 1 and 2 are also (1,2). Similarly, the first two ranks of each pattern in Groups 3 and 4 are (2,1). 
\end {proof}

\begin{theorem}\label {theorem2}
 The last two ranks of each pattern in Groups 1 and 3 are (1,2), and those of each pattern in Groups 2 and 4 are (2,1).
\end{theorem}
\begin{proof}
Using the same method as for Theorem \ref{theorem1}, we can see that the last two ranks of each pattern in Groups 1 and 3 are (1,2), and those of each pattern in Groups 2 and 4 are (2,1).
\end {proof}

\begin{theorem}\label {theorem3}
The GP-Fusion strategy is complete. 
\end{theorem}
\begin{proof}
Proof by contradiction. We assume that pattern \textbf{p} is in Groups 1 or 3, pattern \textbf{q} is in Groups 3 or 4, and \textbf{p} can fuse with \textbf{q}. From Theorem \ref{theorem2}, we know that the last two ranks of pattern  \textbf{p} are (1,2). According to Theorem \ref{theorem1}, we know that the first two ranks of pattern \textbf{q} are (2,1). Since (1,2) is not equal to (2,1), \textbf{p} cannot fuse with  \textbf{q}. This contradicts the assumption that \textbf{p} can fuse with  \textbf{q}. Similarly, if pattern \textbf{p} is in Groups 2 or 4, and \textbf{q} is in Groups 1 or 2, then \textbf{p} cannot fuse with \textbf{q}. Hence, the GP-Fusion strategy is complete.
\end {proof}
The main advantage of this approach is as follows. We know that GP-Fusion divides all patterns into four groups based on their characteristics. Each pattern can only be fused with patterns in two groups, rather than four groups. Hence, GP-Fusion only needs to carry out $l^2/2$ pattern fusions rather than $l^2$, which can effectively improve the efficiency of pattern fusion.

\subsubsection{Maximal support priority strategy}

To improve the efficiency of generating candidate patterns using a heuristic approach, we propose a maximal support priority strategy. 

\textbf{Maximal support priority strategy.} We sort the frequent patterns in $F_m$  in descending order according to their supports with the forgetting mechanism to form $Plist$. Then, we select each pattern \textbf{p} from $Plist$ in turn. We create $Qlist$ to store the patterns that can be fused with pattern \textbf{p} using GP-Fusion. Finally, we fuse pattern \textbf{p} with each pattern in $Qlist$.

\begin{theorem}\label {theorem4}
The maximal support priority strategy is complete. 
\end{theorem}
\begin{proof}
We know that the maximal support priority strategy does not remove any patterns, and only adjusts the fusion order of the patterns. It does not change the mining results. Hence, the maximal support priority strategy is complete.
\end {proof}

We consider Example \ref{example4} to illustrate the specific process of GP-Fusion with the maximal support priority strategy.

\begin{example}\label{example4}
 Given the same time series \textbf{t} as in Example \ref{example2}, and an OPF set $F_3 =\{(3,2,1), (2,1,3), (1,3,2)\}$, the support with forgetting mechanism for each pattern is shown in Table \ref{OPFs with length three}. 

 \begin{table} [h]
\centering
    % \vspace{-1.2em}
    \scriptsize
       \caption{OPFs with length three}
    \begin{tabular}{lcc}
        \toprule
        %\textbf{Type}
         \textbf{p} &\textbf{$fsup(\textbf{p}, \textbf{t})$}&\textbf{Group}\\
        \midrule
        (3,2,1)& 2.11&4\\
        (2,1,3)& 1.51&3\\
        (1,3,2)& 2.17&2\\
         \bottomrule
    \end{tabular}
	\label{OPFs with length three}
    % \vspace{-0.3cm}
\end{table}

Step 1: Sort the patterns in $F_3$ in descending order of $fsup(\textbf{p}, \textbf{t})$ to obtain $Plist=\{(1,3,2), (3,2,1), (2,1,3)\}$, as shown in Table \ref{in descending order}.

 \begin{table} [h]
\centering
    % \vspace{-1.2em}
    \scriptsize
       \caption{$F_3$ in descending order}
    \begin{tabular}{lcc}
        \toprule
        %\textbf{Type}
         \textbf{p} &\textbf{$fsup(\textbf{p}, \textbf{t})$}&\textbf{Group}\\
        \midrule
        (1,3,2)& 2.17&2\\
        (3,2,1)& 2.11&4\\
        (2,1,3)& 1.51&3\\
     
         \bottomrule
    \end{tabular}
	\label{in descending order}
    % \vspace{-0.3cm}
\end{table}

Step 2: From Table \ref{in descending order}, we first select pattern \textbf{p}=(1,3,2), which belongs to Group 2. Following the GP-Fusion strategy, we select patterns in Groups 3 and 4 from $Plist$ to obtain $Qlist=\{(3,2,1), (2,1,3)\}$. Finally, we fuse \textbf{p} with each pattern in $Qlist$ to generate candidate patterns (1,4,3,2), (2,4,3,1), and (1,3,2,4).
 
Step 3:  We iterate Step 2 to get the set of candidate patterns with length four $C_4$=$\{$(1,4,3,2), (2,4,3,1),(1,3,2,4),(4,3,2,1), (3,2,1,4),(4,2,1,3),(2,1,4,3), (3,1,4,2)$\}$.

\end{example}

\subsection{Support calculation with forgetting mechanism}
\label{Support calculation with forgetting mechanism}

In this section, we propose an SCF algorithm with two pruning strategies to calculate the support. 

To utilize the matching results of the sub-patterns, SCF calculates the occurrences of super-patterns with length $m+1$ based on the occurrences of sub-patterns with length $m$. The occurrences of pattern \textbf{p} are stored in $L_p$. To avoid repeatedly calculating the forgetting value for each position, we calculate and store them in an array \textbf{f}. SCF consists of the following four steps.

Step 1: Calculate the forgetting values for each position in a time series using Equation \ref{equation1} and store them in array \textbf{f}.

 \begin{example}
We adopt the same time series \textbf{t} as in Example \ref{example2}, for which the forgetting values are shown in Table \ref{Forgetting values for time series}.

\begin{table} [h]
\centering
\tiny
    % \vspace{-1.2em}
    \scriptsize
       \caption{Forgetting values for time series \textbf{t}}
    \begin{tabular}{lcccccccccc}
        \toprule
        %\textbf{Type}
         \textbf{ID} &\textbf{1} &\textbf{2}&\textbf{3}&\textbf{4}&\textbf{5}&\textbf{6}&\textbf{7}&\textbf{8}&\textbf{9}&\textbf{10}\\
        \midrule
        \textbf{f}&   0.41	&0.45 &0.50&	0.55&	0.61&	0.67&	0.74  &	0.82&	0.90 &	1.00\\
               \bottomrule
    \end{tabular}
	\label{Forgetting values for time series}
    % \vspace{-0.3cm}
\end{table}
 
 \end{example}
 
Step 2: Store $L_\textbf{p}$ in the prefix set $pre_\textbf{p}$ and suffix set $suf_\textbf{p}$, and $fsup(\textbf{p},\textbf{t})$ in the suffix support set $sufsup_\textbf{p}$.
 
Step 3: Scan sets $pre_\textbf{p}$ and $suf_\textbf{q}$ to generate the 
 occurrence sets $L_\textbf{r}$ and $L_\textbf{h}$ of super-patterns \textbf{r} 
 and \textbf{h}, which can be generated by $\textbf{p}\oplus\textbf{q}$.
 Assuming $<$$i$$>$ $\in pre_\textbf{p}$ and $<$$j$$>$$ \in suf_\textbf{q}$, 
 SCF is used to generate the occurrence sets $L_\textbf{r}$ and $L_\textbf{h}$ according to the following rules.
 
 Rule 1: According to Definition \ref{definition8}, if $p_1=q_m$, then $\textbf{r}, \textbf{h}=\textbf{p}\oplus\textbf{q}$.

 If $i+1=j$, then $<$$j$$>$ may be an occurrence of \textbf{r} or \textbf{h}. We then compare the values of $t_{first}$ and $t_{last}$ in time series \textbf{t}, where $first=j-m$ and $last=j$.

1)	If $t_{first}<t_{last}$, then $<$$j$$>$ is an occurrence of \textbf{r}, i.e., $j \in L_\textbf{r}$, and $<$$i$$>$ and $<$$j$$>$ are removed from $pre_\textbf{p}$  and $suf_\textbf{q}$, respectively.

2)	If $t_{first}>t_{last}$, then $<$$j$$>$ is an occurrence of \textbf{h}, i.e., $j \in L_\textbf{h}$, and $<$$i$$>$ and $<$$j$$>$ are removed from $pre_\textbf{p}$  and $suf_\textbf{q}$, respectively.

3)	If $t_{first}=t_{last}$, then $<$$j$$>$ is not an occurrence of \textbf{r} or \textbf{h}.

Rule 2: According to Definition \ref{definition8}, if $p_1\neq q_m$, then $\textbf{r}=\textbf{p}\oplus\textbf{q}$.

If $i+1=j$, then $<$$j$$>$ is an occurrence of \textbf{r}, i.e., $j \in L_\textbf{r}$, and $<$$i$$>$ and $<$$j$$>$ are removed from $pre_\textbf{p}$  and $suf_\textbf{q}$ respectively.

Step 4: Update $fsup(\textbf{r},\textbf{t}) $ and  $sufsup_\textbf{q}$. 
If  $<$$j$$>$ is an occurrence of pattern \textbf{r}, then $fsup(\textbf{r},\textbf{t})+=f_j$ and 
$sufsup_\textbf{q}-=f_j$, since $<$$j$$>$ is an occurrence of pattern \textbf{q}.
Note that $fsup(\textbf{r},\textbf{t})$ is related to $<$$j$$>$ but not $<$$i$$>$. Hence, we only need to update $sufsup_\textbf{q}$, and not $presup_\textbf{p}$.

We use Example \ref{example6} to demonstrate the process of SCF.

\begin{example}\label{example6}
Consider the same time series \textbf{t} as in Example \ref{example2}. The matching sets of \textbf{p}=(1,3,2) and \textbf{q}=(3,2,1) are $L_\textbf{p}=\{3,6,10\}$ and $L_\textbf{q}=\{4,7,8\}$, respectively.

From Example \ref{example3}, we know that $ \textbf{p}\oplus \textbf{q}$ can generate two candidate patterns, \textbf{r}=(1,4,3,2) and \textbf{h}=(2,4,3,1). We know $pre_\textbf{p}=L_\textbf{p}=\{3,6,10\}$, $|pre_\textbf{p}|=3$, $suf_\textbf{q}=L_\textbf{q}=\{4,7,8\}$, and 
$sufsup_\textbf{q}=f_4+f_7+f_8=0.55+0.74+0.82=2.11$. It can be seen that $3 \in  pre_\textbf{p}$, and $3+1=4 \in suf_\textbf{q}$. 
Hence, according to Rule 1, $<$4$>$ may be an occurrence of \textbf{r} or \textbf{h}.
Since $t_{first}=t_{4-3}=t_1=15$, $t_{last}=t_4=27$, and $15<27$, $<$4$>$ is an occurrence of \textbf{r}. We then add $<$4$>$ to $L_\textbf{r}$, and remove $<$3$>$ and $<$4$>$ from $pre_\textbf{p}$ and $suf_\textbf{q}$, respectively. 
Moreover, $fsup(\textbf{r},\textbf{t})=f_4 =0.55$, and $sufsup_\textbf{q}=2.11-f_4=2.11-0.55=1.56$. Similarly, the matching set of \textbf{h} is $L_\textbf{h}={7}$, and $<$6$>$ and $<$7$>$ are removed from $pre_\textbf{p}$ and $suf_\textbf{q}$ respectively, i.e., $pre_\textbf{p}=\{10\}$ and $suf_\textbf{q}=\{8\}$, $fsup(\textbf{h},\textbf{t})=f_7=0.74$, and $sufsup_\textbf{q}=1.56-0.74=0.82$.

\end{example}

We also propose a prefix pattern pruning strategy and a suffix pattern pruning strategy to improve the mining efficiency.

\textbf{Prefix pattern pruning strategy.} If $|pre_\textbf{p}|$$<$\textit {minsup}, then we can prune pattern \textbf{p} and its super-patterns, which means that it is no longer possible to generate OPFs with \textbf{p} as the prefix OPP.

\begin{theorem}\label {theorem5}
The prefix pattern pruning strategy is complete.
\end{theorem}
\begin{proof}
We assume $\textbf{r}=\textbf{p}\oplus\textbf{q}$, where \textbf{p} is the prefix OPP of \textbf{r}, \textbf{q} is the suffix OPP of \textbf{r}, and \textbf{r} is the super-pattern of \textbf{p} and \textbf{q}. The number of occurrences of \textbf{r} is no greater than that of \textbf{p}, i.e., $|L_\textbf{r}|\leq  |pre_\textbf{p}|$, since the OPP mining process has the property of anti-monotonicity. 
If $|pre_\textbf{p}|$$<$\textit {minsup}, then $|L_\textbf{r}|$$<$\textit {minsup}. 
In addition, $fsup(\textbf{r},\textbf{t})=\sum_{j=1}^{n} f_j \leq |L_\textbf{r}|$, since no $f_j$ is greater than one. 
Thus, $fsup(\textbf{r},\textbf{t})$$<$\textit {minsup}, which means that \textbf{r} is not an OPF. Hence, if $|pre_\textbf{p}|$$<$\textit {minsup}, then we can prune \textbf{p} and its super-patterns.
\end {proof}

\textbf{Suffix pattern pruning strategy.} If $sufsup_\textbf{q}$$<$\textit {minsup}, then we can prune pattern \textbf{q} and its super-patterns, which means that it is impossible to generate OPFs with \textbf{q} as the suffix OPP.

\begin{theorem}\label {theorem6}
The suffix pattern pruning strategy is complete.
\end{theorem}
\begin{proof}
The number of occurrences of a super-pattern is no greater than that of its suffix pattern, since OPP mining has the property of anti-monotonicity. Note that the occurrence set of a super-pattern is a subset of that of its suffix pattern. The support of a super-pattern is therefore no greater than that of its suffix pattern. Hence, if $sufsup_\textbf{q}$$<$\textit {minsup}, then we can prune pattern \textbf{q} and its super-patterns.
\end {proof}

Example \ref{example7} illustrates the effectiveness of the two pruning strategies.

\begin{example}\label{example7}
 Suppose we have patterns \textbf{p}=(1,3,2), \textbf{q}=(3,2,1), and \textbf{g}=(2,1,3) and $minsup=1.5$. \textbf{p} is fused with \textbf{q} and \textbf{g} to generate super-patterns. From Example \ref{example6}, we know that $pre_\textbf{p}=\{10\}$ and $sufsup_\textbf{q}=0.82$ after using $\textbf{p}\oplus\textbf{q}$ to generate two candidate patterns. 
 If we do not employ the prefix pattern pruning strategy, we need to use $\textbf{p}\oplus\textbf{g}$ to generate pattern (1,3,2,4) using GP-Fusion, and calculate the support for this pattern. However, \textbf{p} is a prefix pattern that can be pruned, since $|pre_\textbf{p}|=1$$<$\textit {minsup}. Similarly, if we do not employ the suffix pattern pruning strategy, we know from Example \ref{example3} that we need to use $\textbf{q}\oplus\textbf{q}$ to generate pattern (4,3,2,1) and calculate the support for this pattern. However, pattern \textbf{q} is a suffix pattern that can be pruned, since $sufsup_\textbf{q}=0.82$$<$\textit {minsup}. Hence, the prefix and the suffix pattern pruning strategies can further improve the mining efficiency.

\end{example}

When SCF is used to calculate the support of a pattern, the matched elements are dynamically deleted. The advantage of the SCF algorithm is that with a reduction in the number of elements in $pre_\textbf{p}$ and $suf_\textbf{q}$, the number of matches between elements will be greatly reduced, which can avoid redundant calculations. Two pruning strategies are also proposed to further prune candidate patterns. Thus, the efficiency of the method is improved.

Pseudocode for SCF is given in Algorithm \ref{Algorithm 1}.

\begin{algorithm}[htb]	
\caption{SCF}
\begin{algorithmic}[1]\label{Algorithm 1}
\REQUIRE  patterns \textbf{p} and \textbf{q}, $pre_\textbf{p}$, $suf_\textbf{q}$, forgetting values \textbf{f}, and $sup(\textbf{q},\textbf{t})$
\ENSURE  super-patterns \textbf{r} and \textbf{h} and corresponding $L_\textbf{r}$ and $L_\textbf{h}$
 \STATE {$L_\textbf{r} \leftarrow L_\textbf{h} \leftarrow \{\}$;}
 \IF {\textit {suffixop} (\textbf{p})==\textit {prefixop}(\textbf{q})}
\FOR {each occurrence $i$ in $pre_\textbf{p}$}
   \STATE {$j \leftarrow i+1$;}
\IF {$j$ in $suf_\textbf{q}$ }
\STATE {Remove $i$ and $j$ from $pre_\textbf{p}$ and $suf_\textbf{q}$, respectively, and add $j$ into $L_\textbf{r}$ or $L_\textbf{h}$;}
\STATE {$sup(\textbf{q},\textbf{t}) \leftarrow sup(\textbf{q},\textbf{t})-f_j$;}
	 \ENDIF   
\ENDFOR
    \ENDIF
\RETURN $\textbf{r}$, $\textbf{h}$, $L_\textbf{r}$, $L_\textbf{h}$
\end{algorithmic}
\end{algorithm}

\subsection{OPF-Miner}
\label{OPF-Miner}
In this section, we describe OPF-Miner, a method of mining all frequent OPFs in which GP-Fusion is applied with the maximal support priority strategy to generate candidate patterns, and the SCF algorithm with the prefix and suffix pattern pruning strategies to calculate the support with the forgetting mechanism. The main steps of OPF-Miner are as follows.

Step 1: Scan time series \textbf{t} to calculate the forgetting values and store them in array \textbf{f}.

Step 2: Calculate the matching sets and the supports with the forgetting mechanism of patterns (1,2) and (2,1) in time series \textbf{t}. If a pattern is frequent, it is stored in the frequent pattern set $F_2$  with the corresponding group label.

Step 3: Sort $F_m$ in descending order to obtain $Plist$. Select each pattern in $Plist$ to obtain its $Qlist$, and generate the super-patterns using the maximal support priority strategy. Then, use the SCF algorithm with pruning strategies to calculate the support of the super-patterns. If a pattern is frequent, then store it in set $F_{m+1}$  with the corresponding group label.

Step 4: Iterate Step 3 until no super-patterns can be generated. Finally, all OPFs $UF=F _2 \cup F_3 \cup \cdots  \cup F_m$  are obtained.

Example \ref{example8} illustrates the mining process of OPF-Miner.

\begin{example}\label{example8}
We consider the same time series \textbf{t} as in Example \ref{example2}, and set $minsup=1.5$. We mine all frequent OPFs as follows.

Step 1: Scan time series \textbf{t} to calculate the forgetting value of each position, and store it in array \textbf{f}, as shown in Table \ref{Forgetting values for time series}.

Step 2: Calculate the matching sets of (1,2) and (2,1): $L_{(1,2)}$ $=\{2,5,9\}$ and $L_{(2,1)} =\{3,4,6,7,8,10\}$. Then use array \textbf{f} to calculate their supports with the forgetting mechanism. The support of (1,2) with the label Group 1 is 1.96, and that of (2,1) with the label Group 2 is 4.28. The two patterns are stored in $F_2=\{(1,2), (2,1)\}$, since they are frequent.

Step 3: Sort $F_2$ in descending order to obtain $Plist=\{(2,1), (1,2)\}$. First, select pattern $p=(2,1)$ with label Group 2. We obtain $Qlist=\{(2,1), (1,2)\}$ according to Table \ref{Fusion results for patterns with length two}. We then select pattern $q=(2,1)$ as the first pattern in $Qlist$ to generate the super-pattern $(3,2,1)=(2,1) \oplus (2,1)$, since $|pre_{(2,1)}|=$ $|L_{(2,1)}|=6$$>$\textit {minsup} and 
$sufsup_{(2,1)} =fsup(2,1)=4.28$$>$\textit {minsup}. 
We use SCF to obtain $L_{(3,2,1)}= \{4,7,8\}$ and $fsup(3,2,1)=2.11$$>$\textit {minsup}. 
Thus, pattern (3,2,1) is a frequent OPF, which is stored in $F_3$  with the label Group 4 according to Table \ref{Fusion results for patterns with length two}. 
Now, $pre_{(2,1)}=\{4,8,10\}$, $suf_{(2,1)} =\{3,6,10\}$ and $sufsup_{(2,1)}=4.28-2.11=2.17$. 
Next, we select (1,2) as the suffix OPP. It is easy to see that (1,2) and (2,1) cannot be pruned. Thus, $(2,1) \oplus (1,2)=(2,1,3), (3,1,2)$ are generated. SCF calculates $L_{(2,1,3)} =\{5,9\}$. Thus, $fsup(2,1,3)=1.51$, which means that (2,1,3) is frequent, and it is therefore stored in $F _3$  with the label Group 3. Moreover, $fsup(3,1,2)=0$, which means that (3,1,2) is infrequent and is pruned.

Next, pattern (1,2) is selected from $Plist$. In the same way, the frequent pattern (1,3,2) with the label Group 2 can be obtained. Finally, the set of frequent patterns with length three is $F_3 =\{(3,2,1), (2,1,3), (1,3,2)\}$.

Step 4: Iterate Step 3 to generate candidate patterns with length four, as in Examples \ref{example4} and \ref{example7}. It is easy to see that all of the candidate patterns with length four are infrequent, i.e., $F_4 =\emptyset$. Hence, the full set of OPFs is $F=F_2 \cup F_3 \cup F_4=\{(1,2), (2,1), (3,2,1), (2,1,3), (1,3,2)\}$. 	
\end{example}

Pseudocode for OPF-Miner is given in Algorithm \ref{Algorithm 2}.

\begin{algorithm}[htb]	
\caption{OPF-Miner}
\begin{algorithmic}[1]\label{Algorithm 2}
\REQUIRE Time series \textbf{t} and $minsup$
\ENSURE  Frequent OPF set $F$
 \STATE {Calculate the forgetting values \textbf{f};}
 \STATE {Scan time series \textbf{t} to obtain the occurrences of patterns (1,2) and (2,1), which are stored in $L_{(1,2)}$ with the label Group 1 and $L_{(2,1)}$ with the label Group 2, \!respectively;}
 \STATE {Calculate the supports of patterns (1,2) and (2,1) according to \textbf{f}, and store the frequent OPFs in $F_2$;}
  \STATE {$m \leftarrow 2$;}
  \WHILE{$F_m <> NULL$}
 \STATE { Sort $F_m$ in descending order to obtain $Plist$};
\FOR {each pattern \textbf{p} in $Plist$}
   \STATE {$pre_\textbf{p} \leftarrow suf_\textbf{p} \leftarrow L_\textbf{p}$;}
\ENDFOR
\FOR {each pattern \textbf{p} in $Plist$}
   \STATE {Generate $Qlist$ according to \textbf{p};}
   \FOR {each pattern \textbf{q} in $Qlist$}
   \STATE {\textbf{r},\textbf{h},$L_\textbf{r}$,$L_\textbf{h} \leftarrow SCF(\textbf{p},\textbf{q},pre_\textbf{p},suf_\textbf{q},\textbf{f}, sup(\textbf{q},\textbf{t}))$;}
   \STATE {Calculate the supports of \textbf{r} and \textbf{h} according to \textbf{f}, find frequent OPFs, and store them in $F_{m+1}$ with the corresponding group label;}
\ENDFOR
\ENDFOR
\STATE {$m \leftarrow m+1$;}
\ENDWHILE
\STATE {$F=F _2 \cup F_3 \cup \cdots  \cup F_m$;}
\RETURN $F$
\end{algorithmic}
\end{algorithm}

\subsection{Time and space complexity analysis}\label {complexity}
\begin{theorem}\label {theorem8}
The time complexity of OPF-Miner is $O(L \times (n+L))$, where  $n$ and $L$ are  the length of the time series and the number of candidate patterns, respectively.
\end{theorem}
\begin{proof}
The time complexity of calculating the corresponding forgetting values is $O(n)$, since the length of the time series is $n$. The time complexity of Lines 2 and 3 is also $O(n)$. The time complexity of Line 6 is $O(L \times log (L))$. The time complexity of generating candidate patterns is $O(L \times L)$. The time complexity of SCF is $O(n)$. Thus, the time complexity of calculating supports of all candidate patterns is $O(L \times n)$. Hence, the time complexity of OPF-Miner is $O(n+L \times log (L)+L \times L+ L \times n)$ =$O(L \times (n+L))$.
\end {proof}

\begin{theorem}\label {theorem7}
The space complexity of OPF-Miner is $O(m \times (n+ L))$, where $m$ is the length of the maximum OPF.
\end{theorem}
\begin{proof}
The length of the time series is $n$. Thus, the space complexity of the corresponding forgetting values is $O(n)$. The space complexity of $L_{(1,2)}$ and $L_{(2,1)}$ is also $O(n)$. Moreover, the size of all occurrences of all patterns with length $j$ is $n$, since each position only belongs to an occurrence of a pattern. Thus, the space complexity of all occurrences of all patterns with length $j$ is $O(n)$. Therefore,  the space complexity of all occurrences of all frequent OPFs is $O(m \times n)$. Moreover, the space complexity of all candidate patterns is $O(m \times L)$.  Hence, the space complexity of OPF-Miner is $O(n+m \times n+ m \times L)$ =$O(m \times (n+L))$.
\end {proof}

\section{Experimental Results and Analysis}
\label{section:Experimental Results and Analysis}

In relation to the performance of the OPF-Miner 
algorithm, we propose the following research questions (RQs):

\begin{enumerate}[RQ1:]
\item How does the GP-Fusion strategy perform compared with the enumeration strategy and the pattern fusion method?
\item What is the performance of the maximal support priority strategy?
\item Can SCF improve the efficiency of the support calculation process?
\item Can the prefix and suffix pattern pruning strategies reduce the number of candidate patterns?
\item How do GP-Fusion, the maximal support priority strategy, and the suffix pattern pruning strategy perform, compared with the strategies in EFO-Miner?
%\item  What is the scalability of OPF-Miner?
\item How do different $minsup$ and $k$ values affect the performance of OPF-Miner?
\item  What is the effect of the forgetting mechanism on the mining results?
\end{enumerate}

To address RQ1, we propose the OPF-Enum and OPF-noGroup algorithms to verify the performance of the GP-Fusion strategy. To answer RQ2, we develop OPF-noPriority and OPF-minPriority to explore the effect of the maximal support priority strategy. In response to RQ3, we propose Mat-OPF to verify the efficiency of SCF. To address RQ4, we develop OPF-Same, OPF-noPre, OPF-noSuf, and OPF-noPrune to validate the performance of the pruning strategies. For RQ5, we propose EFO-OPF to validate the overall effect of GP-Fusion, the maximal support priority strategy, and the suffix pattern pruning strategy.  To answer RQ6, we verify the running performance with different values of $minsup$ and $k$. For RQ7, we consider OPP-Miner and EFO-Miner to investigate the effect of the forgetting mechanism. 

%In response to RQ6, we test the scalability of OPF-Miner on datasets with different lengths.

\subsection{Datasets and baseline algorithms}
\label{Datasets and baseline algorithms}

We used the following nine real stock datasets. Moreover, we set $k$=1/$n$ for all the experiments, where $n$ is the length of the time series. DB1 is an Amazon stock price dataset, which can be downloaded from https://tianchi.aliyun.com/. DB2, DB3, DB4, and DB5 are daily stock indexes, which can be downloaded from https://www.yahoo.com/. DB6, DB7, and DB8 are the daily closing prices of three American stocks, which can be downloaded from https://www.datafountain.cn/. DB9 contains the daily closing prices of 400 A-shares, which can be downloaded from https://tushare.pro/. Table \ref{Description of datasets} provides a detailed description of each dataset.

\begin{table} [h]
\centering
    % \vspace{-1.2em}
    \scriptsize
       \caption{Description of datasets}
    \begin{tabular}{lccc}
        \toprule
        %\textbf{Type}
         \textbf{Name}&\textbf{Dataset} &\textbf{Length} &\textbf{Number of sequences}\\
        \midrule
       DB1&Amazon stock&5,842&1\\
       DB2&	Russell 2000&	8,141&	1\\
       DB3&	Nasdaq&	12,279&	1\\
       DB4&	S\&P500	&23,046	&1\\
DB5	&NYSE	&60,000	&1\\
DB6	&Cl.US	&10,305	&1\\
DB7	&Hpq.US	&12,075	&1\\
DB8	&Ge.US	&14,058	&1\\
DB9	&A-shares	&2,180	&400\\

        \bottomrule
    \end{tabular}
	\label{Description of datasets}
    % \vspace{-0.3cm}
\end{table}

All the experiments were conducted on a computer with an AMD Ryzen 7 5800H, 3.2GHz CPU, 16GB RAM, and the Windows 11 64-bit operating system. The compilation environment was IntelliJ IDEA 2022.2.

To investigate the performance of OPF-Miner, we selected 12 competitive algorithms, and all data and algorithms can be downloaded from https://github.com/wuc567/Pattern-Mining/tree/master/OPF-Miner.

\begin{enumerate} [1.]

\item  OPF-Enum and OPF-noGroup. To verify the performance of the GP-Fusion strategy, we developed OPF-Enum and OPF-noGroup, which apply the enumeration strategy and the pattern fusion method to generate candidate patterns, respectively.

%Our problem is far different from the classical sequential

\item OPF-noPriority and OPF-minPriority. To validate the performance of maximal support priority strategy, we developed OPF-noPriority and OPF-minPriority. OPF-noPriority does not use the maximal support priority strategy, and OPF-minPriority applies a minimal support priority strategy, which means the frequent patterns are sorted in ascending order.

\item   Mat-OPF \cite {OPPminer}. To evaluate the efficiency of SCF, we devised Mat-OPF, which uses OPP matching to calculate the support of candidate patterns. 
	
%\textbf{For performance of pattern fusion strategy:}
\item OPF-Same, OPF-noPre, OPF-noSuf, and OPF-noPrune. To investigate the effects of the prefix and suffix pattern pruning strategies, we designed OPF-Same, OPF-noPre, OPF-noSuf, and OPF-noPrune. OPF-Same adopts the same suffix pattern pruning strategy as the prefix pattern pruning strategy, i.e., if $|suf_\textbf{q}|< minsup$, then \textbf{q} is pruned to avoid fusion with other patterns. OPF-noPre does not use a prefix pattern pruning strategy, OPF-noSuf does not use a suffix pattern pruning strategy, and OPF-noPrune does not apply either the prefix or suffix pattern pruning strategies.

\item To evaluate the overall effect of GP-Fusion, the maximal support priority strategy, and the suffix pattern pruning strategy, we developed EFO-OPF, which uses the pattern fusion strategy for candidate pattern generation.

\item  OPP-Miner \cite {OPPminer} and EFO-Miner \cite {OPRminer}. To validate the performance of the forgetting mechanism, we selected OPP-Miner and EFO-Miner, which mine all frequent OPPs without the forgetting mechanism.
\end{enumerate}

\subsection{Performance of OPF-Miner}
\label{Performance of OPF-Miner}

To investigate the performance of OPF-Miner, we considered 10 competitive algorithms: OPF-Enum, OPF-noGroup, OPF-noPriority, OPF-minPriority, Mat-OPF, OPF-Same, OPF-noPre, OPF-noSuf, OPF-noPrune, and EFO-OPF. 

We conducted experiments on DB1–DB8, and set the value of $minsup=4$. Since all the competitive algorithms were complete, the numbers of OPFs mined were the same, i.e., 623, 914, 1,359, 2,259, 28,201, 718, 837, and 771 OPFs were mined from DB1–DB8, respectively. Figs. \ref{Comparison of running time on SDB1–SDB8}-\ref{Comparison of numbers of support calculations} show comparisons of the running time, the memory consumption, the number of candidate patterns, the number of pattern fusions, and the number of support calculations, respectively. Note that Fig. \ref{Comparison of numbers of pattern fusions} does not show the results for OPF-Enum, since it does not use the prefix and suffix OPPs of sub-patterns to generate candidate patterns. Fig. \ref{Comparison of numbers of support calculations} does not show results for Mat-OPF, since it does not use the matching results of the sub-patterns to calculate the support.

\begin{figure}
    \centering
    \includegraphics[width=0.95\linewidth]{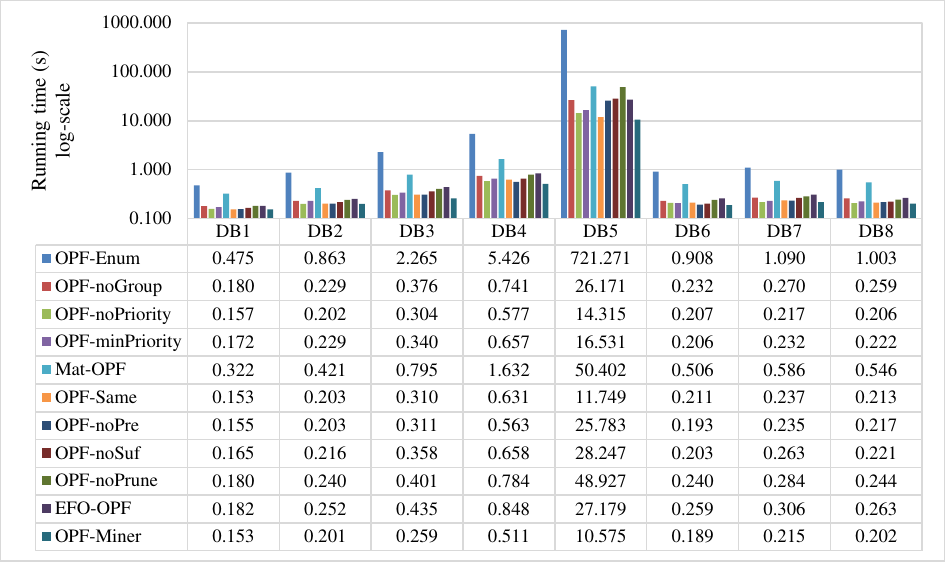}
    \caption{ Comparison of running time on DB1–DB8}
    \label{Comparison of running time on SDB1–SDB8}
\end{figure}

%\vspace{0.5cm}

\begin{figure}
    \centering
   \includegraphics[width=0.95\linewidth]{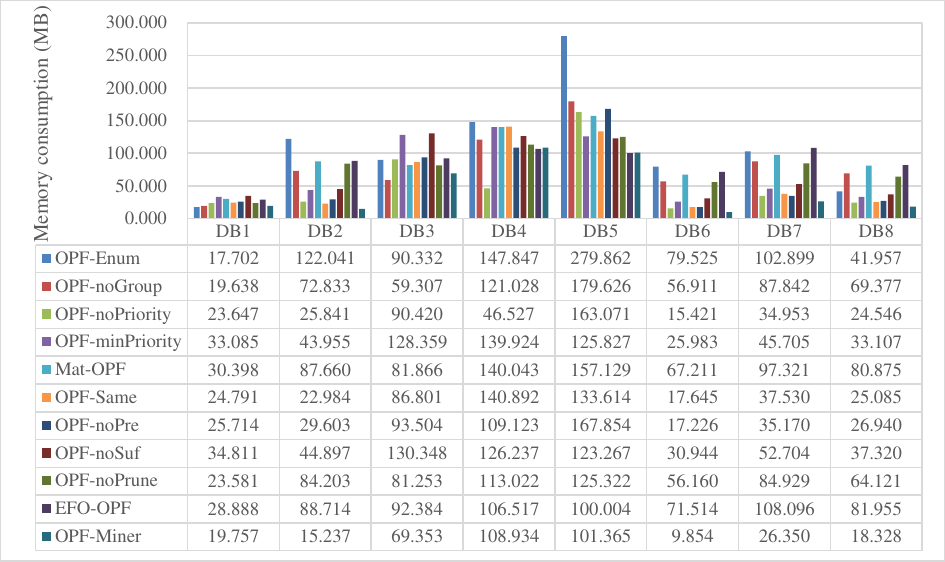}
    \caption{Comparison of memory consumption on DB1–DB8}
    \label{Comparison of memory consumption on DB1–DB8}
\end{figure}

%\vspace{0.5cm}

\begin{figure}
    \centering
    \includegraphics[width=0.95\linewidth]{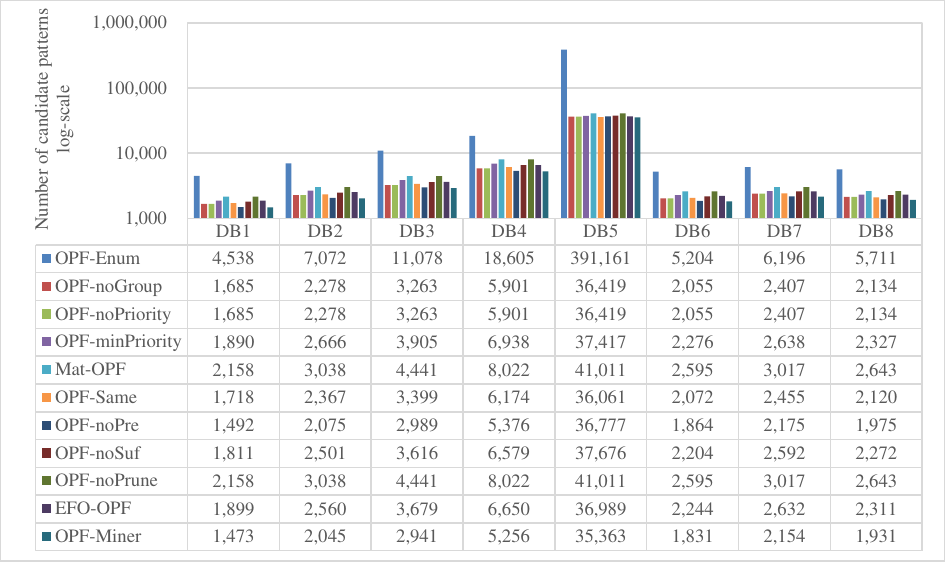}
    \caption{Comparison of numbers of candidate patterns on DB1–DB8}
    \label{Comparison of numbers of candidate patterns on DB1–DB8}
\end{figure}

%\vspace{0.5cm}

\begin{figure}
    \centering
\includegraphics[width=0.95\linewidth]{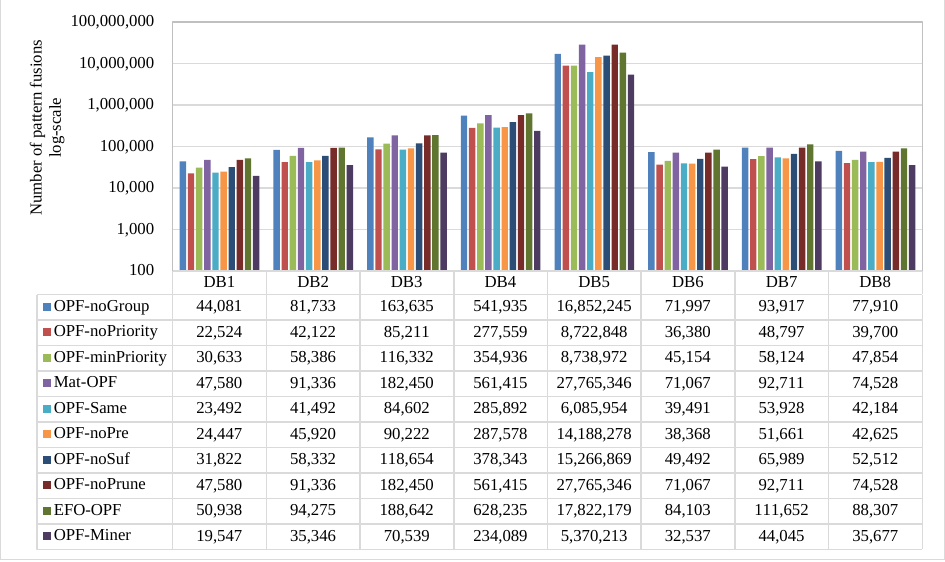}
    \caption{Comparison of numbers of pattern fusions on DB1–DB8}
    \label{Comparison of numbers of pattern fusions}
\end{figure}

%\vspace{0.5cm}

\begin{figure}
    \centering
\includegraphics[width=0.95\linewidth]{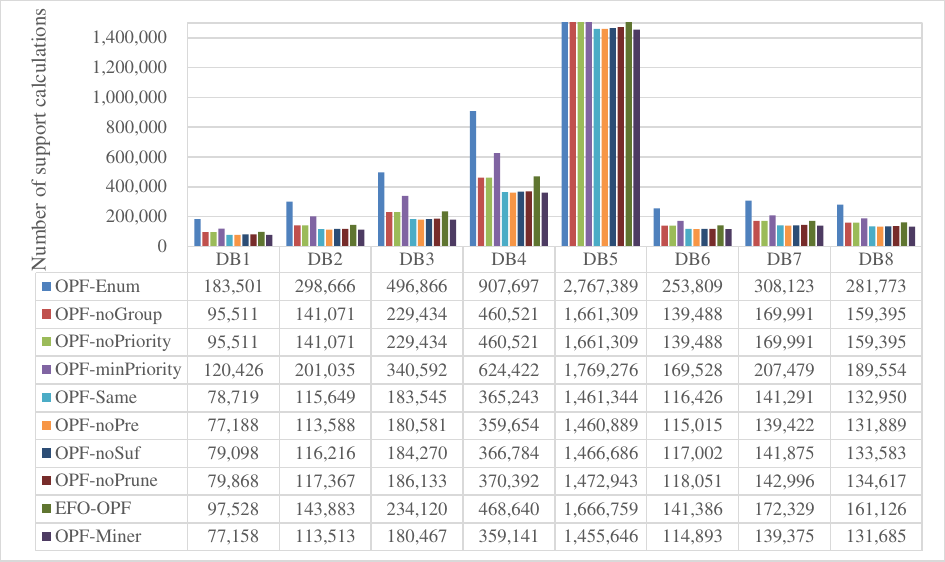}
    \caption{Comparison of numbers of support calculations on DB1–DB8}
    \label{Comparison of numbers of support calculations}
\end{figure}

These results give rise to the following observations.

\begin{enumerate} [1.]
\item OPF-Miner outperforms OPF-Enum, OPF-noGroup, OPF-noPriority, and OPF-minPriority, since OPF-Miner takes less time than the four competitive algorithms. For example, on DB2, Fig. \ref{Comparison of running time on SDB1–SDB8} shows that OPF-Miner takes 0.201s, while OPF-Enum, OPF-noGroup, OPF-noPriority, and OPF-minPriority take 0.863, 0.229, 0.202, and 0.229s, respectively. Moreover, Fig. \ref{Comparison of memory consumption on DB1–DB8} shows that OPF-Miner consumes 15.237MB of memory, while OPF-Enum, OPF-noGroup, OPF-noPriority, and OPF-minPriority consume 122.041, 72.833, 25.841, and 43.955MB, The reasons for this are as follows. In the process of generating candidate patterns, OPF-Miner adopts GP-Fusion with the maximal support priority strategy, which can reduce the number of candidate patterns and pattern fusions, meaning that it can effectively improve the efficiency of pattern fusion. OPF-Enum and OPF-noGroup adopt the enumeration and pattern fusion methods, respectively. Although OPF-noPriority and OPF-minPriority both apply the GP-Fusion strategy, OPF-noPriority does not use the maximal support priority strategy, and OPF-minPriority adopts a minimal support priority strategy. Thus, Fig. \ref{Comparison of numbers of candidate patterns on DB1–DB8} shows that OPF-Miner generates 2,045 candidate patterns, while the four competitive algorithms generate 7,072, 2,278, 2,278, and 2,666 candidate patterns on DB2, respectively. Fig. \ref{Comparison of numbers of pattern fusions} shows that OPF-Miner conducts 35,346 pattern fusions, much fewer than OPF-noGroup, OPF-noPriority, and OPF-minPriority. From Fig. \ref{Comparison of numbers of support calculations}, we know that OPF-Miner conducts 113,513 support calculations, whereas the other algorithms conduct 298,666, 141,071, 141,071, and 201,035, respectively. Hence, the performance of OPF-Miner is superior to that of OPF-Enum, OPF-noGroup, OPF-noPriority, and OPF-minPriority, which demonstrates the superiority of the GP-Fusion and maximal support priority strategies.

\item The performance of OPF-Miner is superior to that of Mat-OPF, since OPF-Miner runs faster and uses less memory than Mat-OPF. For example, on DB4, Fig. \ref{Comparison of running time on SDB1–SDB8} shows that OPF-Miner takes 0.511s, while Mat-OPF takes 1.632s. Fig. \ref{Comparison of memory consumption on DB1–DB8}  shows that OPF-Miner consumes 108.934MB of memory, while Mat-OPF consumes 140.043MB. The reasons for this are as follows. OPF-Miner employs SCF with pruning strategies to calculate the support, fully utilizing the matching results of sub-patterns and pruning the infrequent patterns, meaning that the efficiency is improved. However, Mat-OPF adopts OPP matching, which involves scanning the database from beginning to end for each support calculation. This generates numerous redundant candidate patterns without pruning strategies. Fig. \ref{Comparison of numbers of candidate patterns on DB1–DB8} shows that Mat-OPF generates 8,022 candidate patterns, more than the 5,256 generated by OPF-Miner. Fig. \ref{Comparison of numbers of pattern fusions} shows that Mat-OPF conducts 561,415 pattern fusions, more than twice the value of 234,089 for OPF-Miner. Hence, OPF-Miner gives better performance than Mat-OPF, which demonstrates the efficiency of SCF with two pruning strategies.

\item OPF-Miner gives better performance than OPF-Same, OPF-noPre, OPF-noSuf, and OPF-noPrune. For example, on DB6, Fig. \ref{Comparison of running time on SDB1–SDB8} shows that OPF-Miner takes 0.189s, while OPF-Same, OPF-noPre, OPF-noSuf, and OPF-noPrune take 0.211, 0.193, 0.203, and 0.240s, respectively. Fig. \ref{Comparison of memory consumption on DB1–DB8} shows that OPF-Miner consumes 9.854MB of memory, while the other four algorithms consume 17.645, 17.226, 30.944, and 56.160MB, respectively. The reasons for this are as follows. OPF-Miner adopts prefix and suffix pattern pruning strategies, which can more effectively prune candidate patterns. OPF-Same adopts the same pruning strategy for both suffix and prefix patterns, whereas OPF-noPre, OPF-noSuf, and OPF-noPrune do not employ the prefix pruning strategy, the suffix pruning strategy, and the prefix and suffix pruning strategies, respectively, which means that they generate more candidate patterns. Fig. \ref{Comparison of numbers of candidate patterns on DB1–DB8} shows that OPF-Miner only generates 1,831 candidate patterns, a lower total than the other algorithms. Fig. \ref{Comparison of numbers of pattern fusions} also shows that OPF-Miner conducts 32,537 pattern fusions, fewer than the other algorithms, while from Fig. \ref{Comparison of numbers of support calculations}, we see that OPF-Miner conducts 114,893 support calculations, which is also fewer than the other algorithms. Hence, OPF-Miner outperforms OPF-Same, OPF-noPre, OPF-noSuf, and OPF-noPrune, a finding that validates the effectiveness of prefix and suffix pattern pruning strategies.

\item OPF-Miner outperforms EFO-OPF. For example, on DB8, Fig. \ref{Comparison of running time on SDB1–SDB8}  shows that OPF-Miner takes 0.202s, while EFO-OPF takes 0.263s. From Fig. \ref{Comparison of memory consumption on DB1–DB8} , we know that OPF-Miner consumes 18.328MB, while EFO-OPF consumes 81.955MB. The reason for this is that unlike EFO-OPF, OPF-Miner adopts the GP-Fusion strategy, the maximal support priority strategy, and suffix pattern pruning strategy, which reduce the number of redundant calculations. Hence, OPF-Miner provides better performance than EFO-OPF.
\end{enumerate} 

\subsection{Scalability}
\label{subsection: Scalability}

To investigate the scalability of OPF-Miner, we used DB8 as the experiment dataset and set $minsup=9$. Moreover, we created Ge.US\_1, Ge.US\_2, Ge.US\_3, Ge.US\_4, Ge.US\_5, and Ge.US\_6, which are one, two, three, four, five, and six times the size of Ge.US, respectively. Fig. \ref{Comparison of running time for different database size}  shows the comparison of the running time for different dataset size.

\begin{figure}
   \centering
\includegraphics[width=0.9\linewidth]{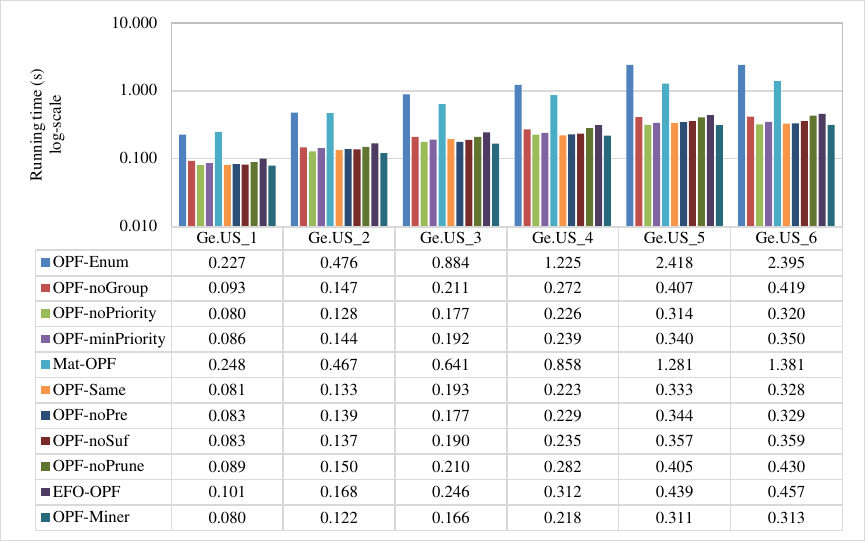}
    \caption{Comparison of running time for different dataset size}
    \label{Comparison of running time for different database size}
\end{figure}

The results give rise to the following observations.

From Fig. \ref{Comparison of running time for different database size}, we can see that the running time of OPF-Miner shows linear growth with the increase of the dataset size. For example, the size of Ge.US\_4 is four times that of Ge.US\_1. OPF-Miner takes 0.080s on Ge.US\_1 and 0.218s on Ge.US\_4, giving 0.218/0.080=2.72. Thus, the growth rate of running time of OPF-Miner is slightly lower than the increase of the dataset size.
This phenomenon can be found on all other datasets. More importantly, we can see that OPF-Miner runs almost at least three times faster than OPF-Enum and Mat-OPF. Therefore, OPF-Miner provides a better scalability than other competitive algorithms. Hence, we draw the conclusion that OPF-Miner has strong scalability, since the mining performance does not degrade as the dataset size increases.

\subsection{Influence of $minsup$}
\label{subsection: Influence of minsup}

To assess the effects of varying $minsup$, we considered OPF-Enum, OPF-noGroup, OPF-noPriority, OPF-minPriority, Mat-OPF, OPF-Same, OPF-noPre, OPF-noSuf, OPF-noPrune, and EFO-OPF as competitive algorithms. Experiments were conducted on DB5 with values of $minusp=$4, 6, 8, 10, and 12. Since all of these algorithms are complete, they mined 28,201, 17,407, 12,083, 9,099, and 7,177 OPFs for these values of $minsup$, respectively. Comparisons of the running time, and the numbers of candidate patterns, pattern fusions, and support calculations are shown in Figs. \ref{Comparison of running time for different values of minsup}–\ref{Comparison of number of support calculations for different values of minsup}, respectively.

\begin{figure}
    \centering  \includegraphics[width=0.95\linewidth]{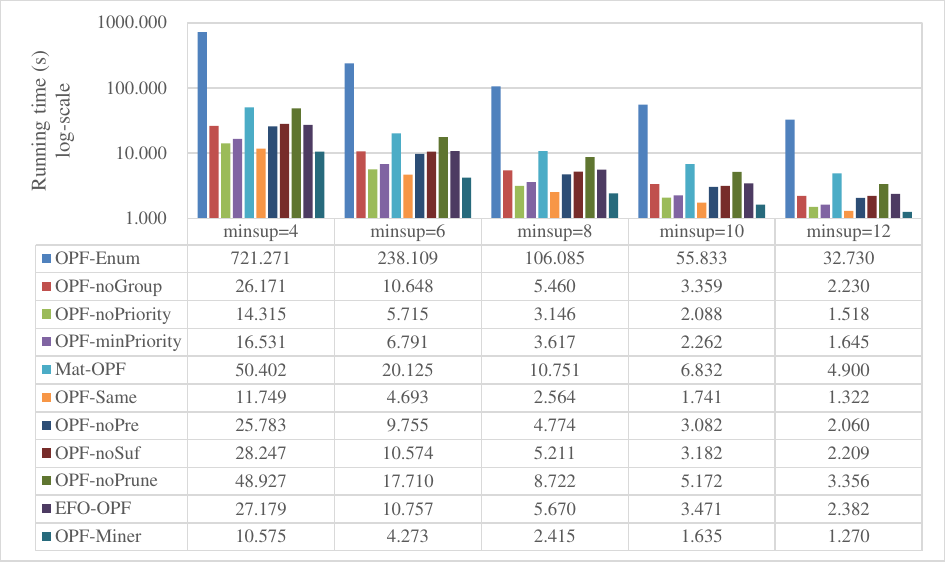}
\caption{Comparison of running time for different values of $minsup$}
    \label{Comparison of running time for different values of minsup}
\end{figure}

\begin{figure}
    \centering
    \includegraphics[width=0.95\linewidth]{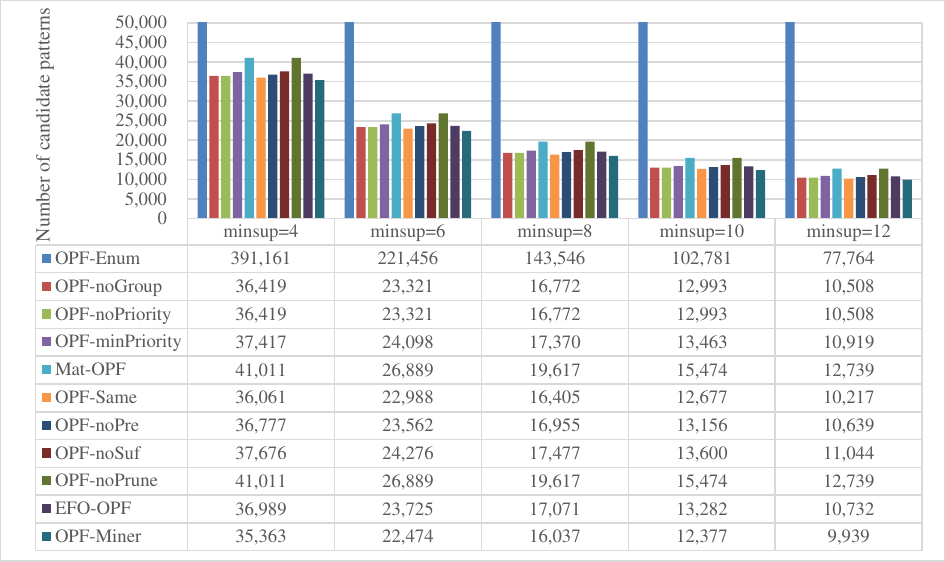}
    \caption{ Comparison of numbers of candidate patterns for different values of $minsup$}
    \label{ Comparison of number of candidate patterns for different values of minsup}
\end{figure}

%\vspace{0.5cm}

\begin{figure}
    \centering
    \includegraphics[width=0.95\linewidth]{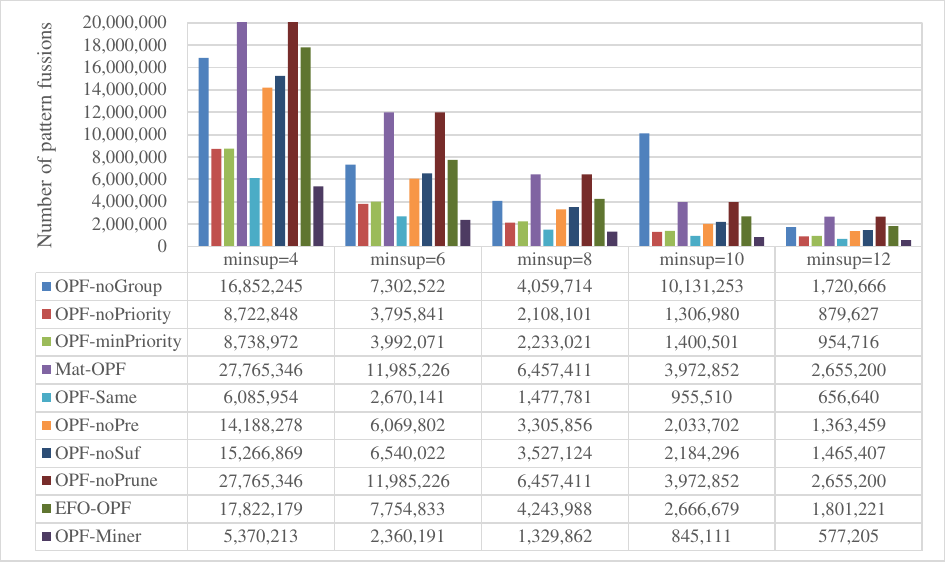}
    \caption{ Comparison of numbers of pattern fusions for different values of $minsup$}
    \label{Comparison of numbers of pattern fusions for different values of minsup}
\end{figure}

%\vspace{0.5cm}

\begin{figure}
    \centering
    \includegraphics[width=0.95\linewidth]{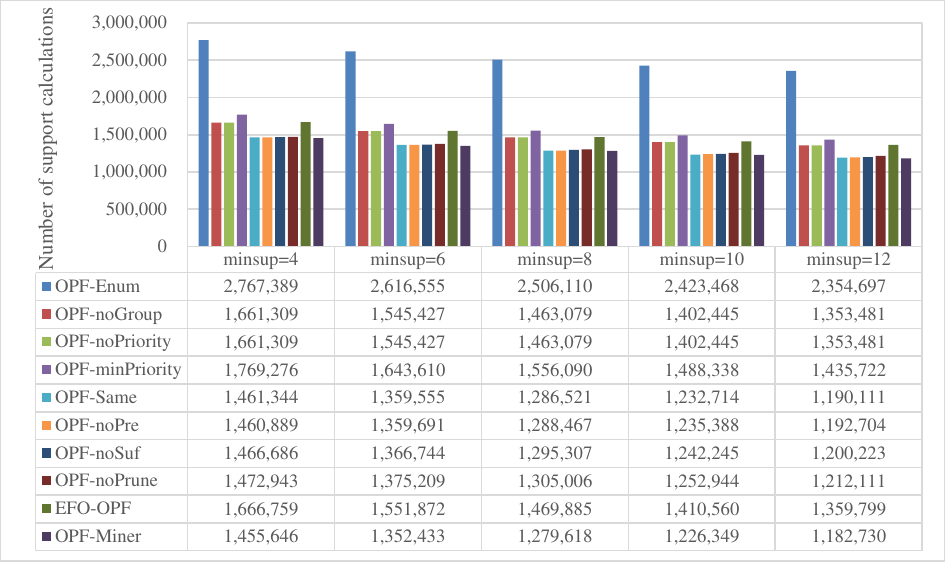}
    \caption{Comparison of numbers of support calculations for different values of $minsup$}
    \label{Comparison of number of support calculations for different values of minsup}
\end{figure}

The following observations can be made from these results.

As $minsup$ increases, there are decreases in the runtime and the numbers of candidate patterns, pattern fusions, and support calculations. For example, for values of $minsup$=10 and $minsup$=12, OPF-Miner takes 1.635 and 1.270s, generates 12,377 and 9,939 candidate patterns, conducts 845,111 and 577,205 pattern fusions, and performs 1,226,349 and 1,182,730 support calculations, respectively. This phenomenon can also be found for the other competitive algorithms. The reasons for this are as follows. With an increase in $minsup$, fewer OPFs can be found, meaning that the algorithm conducts fewer support calculations and fewer pattern fusions. Hence, the algorithm runs faster. More importantly, OPF-Miner outperforms the other competitive algorithms for all values of $minsup$, a finding that supports the results in Section \ref{Performance of OPF-Miner}.

\subsection{Influence of $k$}
To validate the infuence of different $k$, we considered OPF-Enum, OPF-noGroup, OPF-noPriority, OPF-minPriority, Mat-OPF, OPF-Same, OPF-noPre, OPF-noSuf, OPF-noPrune, and EFO-OPF as competitive algorithms. We conducted experiments on DB5 and set $minsup=15$ with values of $k$=$1/7n$, $1/5n$, $1/3n$, $1/n$, $3/n$, $5/n$, and  $7/n$, and there are 8,631, 8,375, 7,744, 5,366, 2,328, 1,352, and 811 OPFs. Figs. \ref{Comparison of running time for different k values}-\ref{Comparison of numbers of support calculations for different k values} show the comparisons of the running time, the numbers of candidate patterns, pattern fusions, and support calculations,respectively.

\begin{figure}
    \centering  \includegraphics[width=0.95\linewidth]{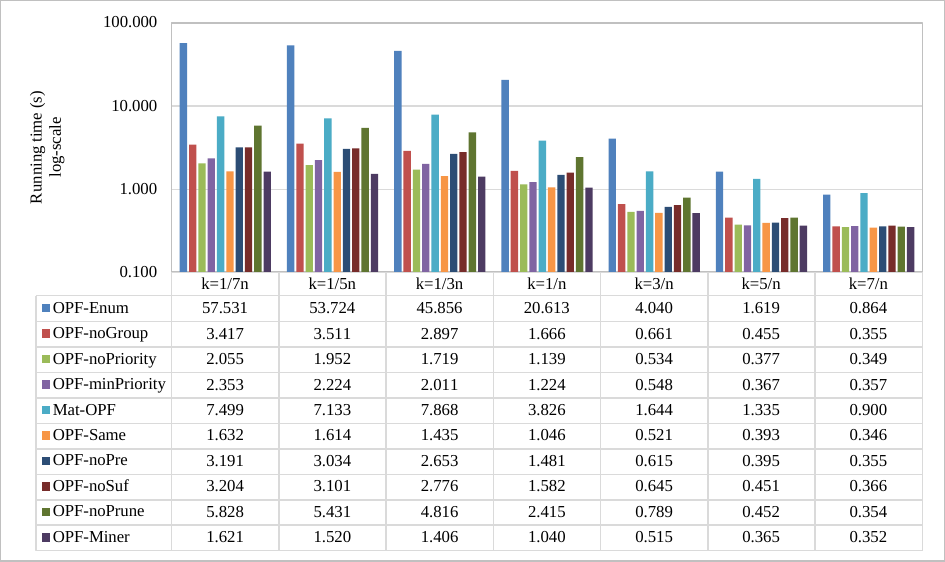}
\caption{Comparison of running time for different values of $k$}
    \label{Comparison of running time for different k values}
\end{figure}

\begin{figure}
    \centering  \includegraphics[width=0.95\linewidth]{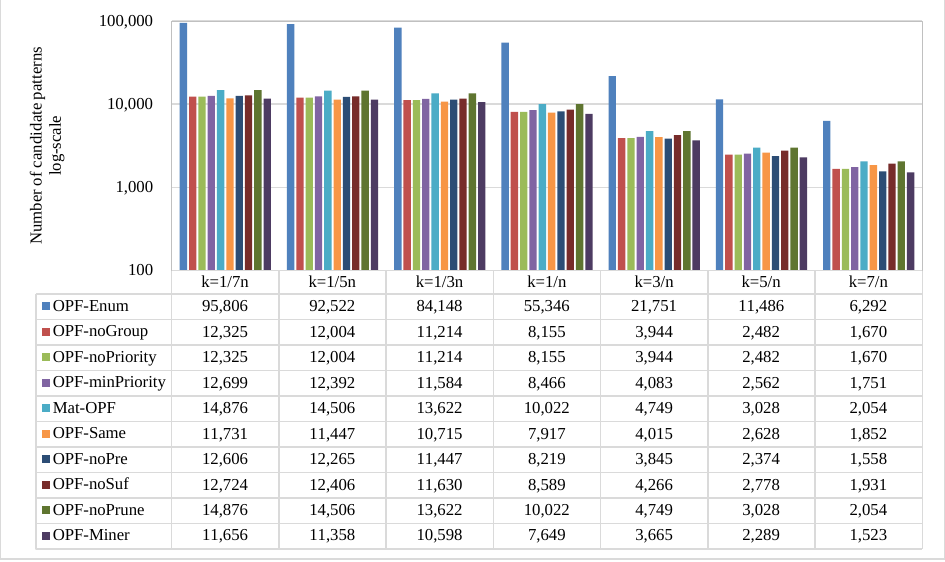}
\caption{Comparison of numbers of candidate patterns for different values of $k$}
    \label{Comparison of numbers of candidate patterns for different k values}
\end{figure}
\begin{figure}
    \centering  \includegraphics[width=0.95\linewidth]{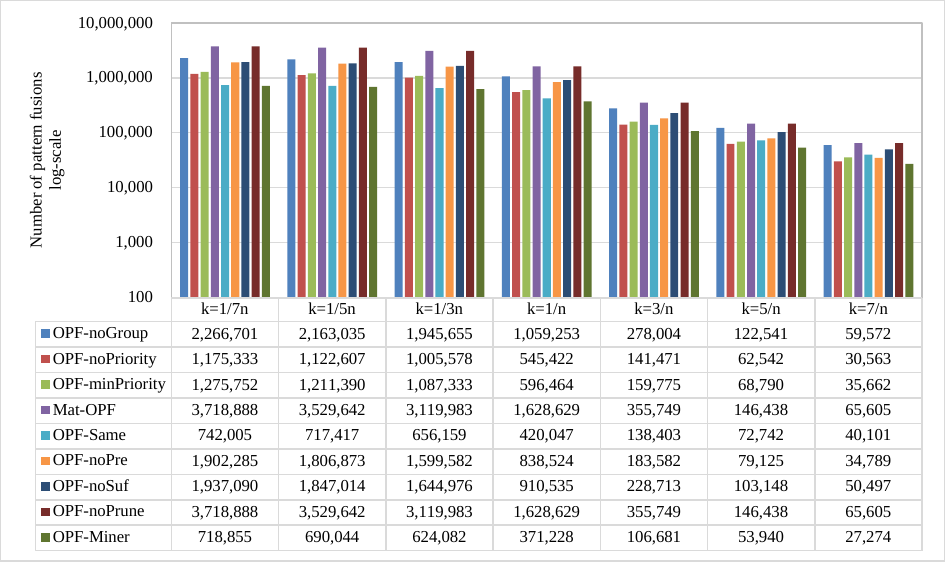}
\caption{Comparison of numbers of pattern fusions for different values of $k$}
    \label{Comparison of numbers of pattern fusions for different k values}
\end{figure}

\begin{figure}
    \centering  \includegraphics[width=0.95\linewidth]{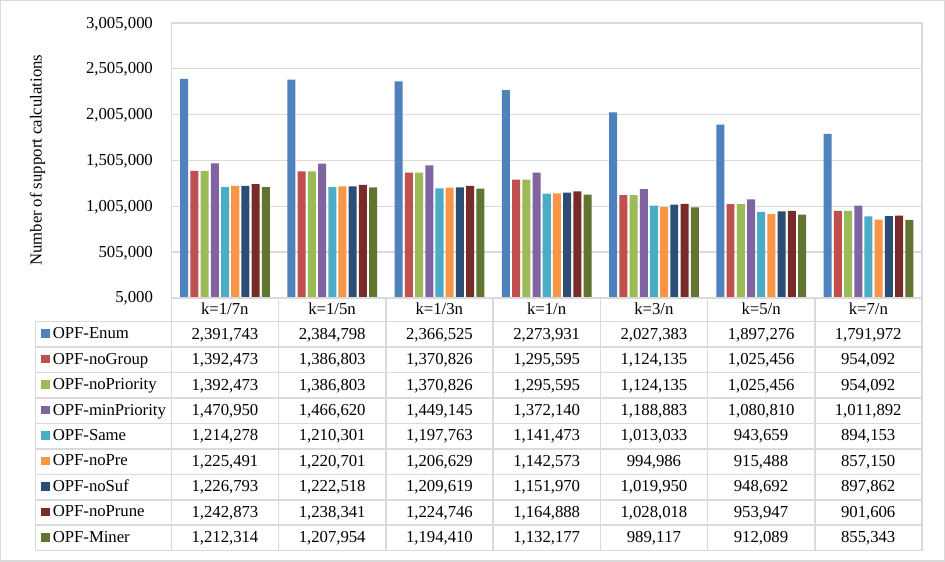}
\caption{Comparison of numbers of support calculations for different values of $k$}
    \label{Comparison of numbers of support calculations for different k values}
\end{figure}

The results give rise to the following observations.

As $k$ increases, there are decreases in the running time, the numbers of candidate patterns, pattern fusions, and support calculations. For example, for $k$=1/$n$ and $k$=3/$n$, OPF-Miner takes 1.040s and 0.515s, generates 7,649 and 3,665 candidate patterns, conducts 371,228 and 106,681 pattern fusions, and performs 1,132,177 and 989,117 support calculations, respectively. This phenomenon can also be found on other competitive algorithms. The reasons are as follows. According to Definition \ref{definition5}, $k$ is the forgetting factor. The larger $k$, the smaller $f_j$, and the smaller $fsup(\textbf{p},\textbf{t})$. In other words, the larger the $k$ value, the smaller the support of the pattern. Therefore, with the same value of $minsup$, fewer OPFs can be found. Thus, with the increase of $k$, the running time, the numbers of candidate patterns, pattern fusions, and support calculations all decrease. More importantly, OPF-Miner outperforms other competitive algorithms for all values of $k$, which supports the results in Section \ref{Performance of OPF-Miner}.

\subsection{Case Study}
To validate the performance of OPF mining in terms of time series feature extraction, we conducted a clustering experiment on DB9 and set $minsup$=50. Silhouette coefficient ($SC$) \cite {sc} and Calinski-Harabasz index ($CHI$) \cite {chi} were used to evaluate the clustering performance.  $SC$ for the whole sample is calculated using Equation \ref{equation2}:

  \begin{equation}
     SC=\frac{\sum_{i=1}^{N} s_i}{N}
     \label{equation2}
    \end{equation}

where $s_i$ is $SC$ of a single sample, and is defined as
 \begin{equation}
     s=\frac{b-a}{max(a,b)}
     \label{equation3}
    \end{equation}

where $a$ represents the average distance between a given sample and the other samples in the same cluster, and $b$ represents the average distance between a given sample and the samples in other clusters.

$CHI $ is calculated using Equation \ref{equation4}:
       \begin{equation}
     CHI=\frac{Tr(B_k)(N-K)}{Tr(W_k)(K-1)}
     \label{equation4}
    \end{equation}               
where $N$ is the number of samples, $K$ is the number of categories, $B_k$  is the inter-class covariance matrix, and $W_k$ is the intra-class covariance matrix.

For $SC$, the closer the value to one the better, whereas for  $CHI$, a larger value is better. To select a suitable forgetting factor $k$, we set the parameter for K-means to $K$=5 and $k$=1/6$n$, 1/4$n$, 1/2$n$, 1/$n$, 2/$n$, 4/$n$, and 6/$n$. The results are shown in Table \ref {differentk}.

\begin{table}
\centering
    % \vspace{-1.2em}
    \tiny
       \caption{Comparison of $SC$ \& $CHI$ on DB9 for different $k$}
    \begin{tabular}{lccccccc}
        \toprule
        %\textbf{Type}
\textbf{}&$k$=1/6$n$& $k$=1/4$n$&$k$=1/2$n$ &$k$=1/$n$&$k$=2/$n$&$k$=4/$n$&$k$=6/$n$\\
        \midrule
$SC$&0.19&0.21 &	0.25& \textbf{0.79}	&0.75&0.69&	0.41 \\
$CHI$&71&	82 &92 &\textbf{760}&	580&	530&310 \\

        \bottomrule
    \end{tabular}
	\label{differentk}
    % \vspace{-0.3cm}
\end{table}

%\end{enumerate}
From Table \ref {differentk}, we know that when $k$=1/$n$, we get the best clustering performance. Hence, in this section, we set $k$=1/$n$. Moreover, we set the parameter for K-means to $K=$3, 4, 5, 6, 7, 8, 9, and 10. We selected two competitive algorithms, OPP-Miner \cite {OPPminer} and EFO-Miner \cite {OPRminer}. We used OPP-Miner and EFO-Miner to mine all frequent OPPs, and OPF-Miner to mine OPFs, with a value of $minsup$=50. The original dataset is referred to in the following as the raw data, denoted as Raw. Figs. \ref{SC} and \ref{CHI} show the clustering results for $SC$ and $CHI$, respectively.

\begin{figure}
    \centering
    \includegraphics[width=0.95\linewidth]{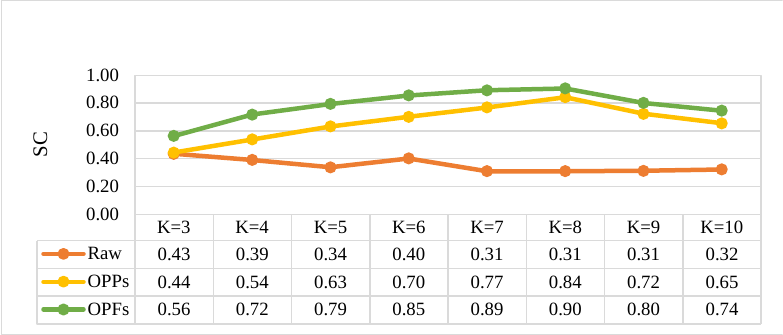}
    \caption{Comparison of $SC$ on DB9}
    \label{SC}
\end{figure}

%\vspace{1.5cm}

\begin{figure}
    \centering
    \includegraphics[width=0.95\linewidth]{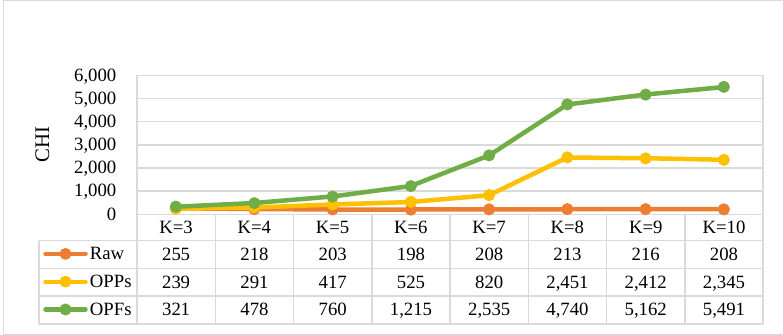}
    \caption{Comparison of \textit{CHI} on DB9 }
     %with different $K$ for K-means
    \label{CHI}
\end{figure}

The results give rise to the following observations.
\begin{enumerate} [1.]
\item The clustering performance using the raw data was the worst, and the OPPs and OPFs gave better performance. This indicates that OPP-Miner, EFO-Miner, and OPF-Miner can effectively extract the critical information from the original time series data. 
For example, from Fig. \ref{SC}, we see that $SC$ for the raw data is 0.34 for $K$=5, while the value for the OPPs and OPFs are 0.63 and 0.79, respectively. Similarly, Fig. \ref{CHI} shows that $CHI$ for the raw data, OPPs and OPFs are 203, 417, and 760, respectively. Hence, the OPPs and OPFs are superior to the raw data. The reason for this is that the raw data contain too much redundant information, which affects the clustering performance, whereas OPP-Miner, EFO-Miner and OPF-Miner can discover useful information from the time series. Hence, OPP-Miner, EFO-Miner and OPF-Miner can effectively be used for feature extraction for the clustering task.

\item OPF-Miner outperforms OPP-Miner and EFO-Miner on feature extraction, i.e., the clustering performance for the OPFs is better than for the OPPs. For example, Figs. \ref{SC} and \ref{CHI} show that the values of $SC$ and $CHI$ for the OPPs are 0.84 and 2,451 for $K$=8, respectively, whereas the values for the OPFs are 0.90 and 4,740, respectively. This is because the OPFs are frequent OPPs calculated with the forgetting mechanism, which assigns different levels of importance to data at different time in the series, whereas the OPPs are based on the assumption that the data at different time are equally important. Hence, we can obtain better clustering performance using OPFs, which demonstrates the effectiveness of the forgetting mechanism.

\end{enumerate}

\section{CONCLUSION}
\label{section:CONCLUSION}
To enable mining of frequent OPPs from data with different levels of importance at different time in a time series, this study has investigated the mining of frequent OPPs with the forgetting mechanism (called OPFs), and has proposed an algorithm called OPF-Miner which performs two tasks: candidate pattern generation and support calculation. For the candidate pattern generation task, we have proposed a method called GP-Fusion that can reduce the number of redundant calculations caused by the pattern fusion strategy. To reduce the number of pattern fusions further, we propose a maximal support priority strategy that prioritizes fusion for patterns with high support. For the support calculation task, we propose an algorithm called SCF that employs prefix and suffix pattern pruning strategies, allowing us to calculate the support of super-patterns based on the occurrence positions of sub-patterns and to prune the redundant patterns effectively. We have also proved the completeness of these pruning strategies. To validate the performance of our approach, nine datasets and 12 competitive algorithms were selected. The experimental results indicated that OPF-Miner provides better performance than other competitive algorithms. More importantly, the features extracted by OPF-Miner can improve the clustering performance for time series.

This paper proposes the OPF-Miner algorithm, which has better performance. Moreover, we have validated that $minsup$ has a greater impact on OPF-Miner performance. However, without sufficient prior knowledge, it is difficult for users to set a more reasonable $minsup$. In the future, we will focus on the study of the top-$k$ order-preserving pattern mining with forgetting mechanism which can discover the top-$k$ OPFs with no need to set $minsup$ in practical applications. Moreover, we know that the forgetting factor can affect the OPF-Miner performance. It is valuable to investigate setting a reasonable forgetting factor without prior knowledge in certain applications.

\section*{Acknowledgement}
This work was partly supported by National Natural Science Foundation of China (62372154, U21A20513, 62076154, 62120106008).

% \bibliographystyle{IEEEtran}

 %\vspace{-1.1cm}
%\vspace{-1.5cm}
\begin{IEEEbiography}[{\includegraphics[width=1.in,height=1.25in,clip,keepaspectratio]{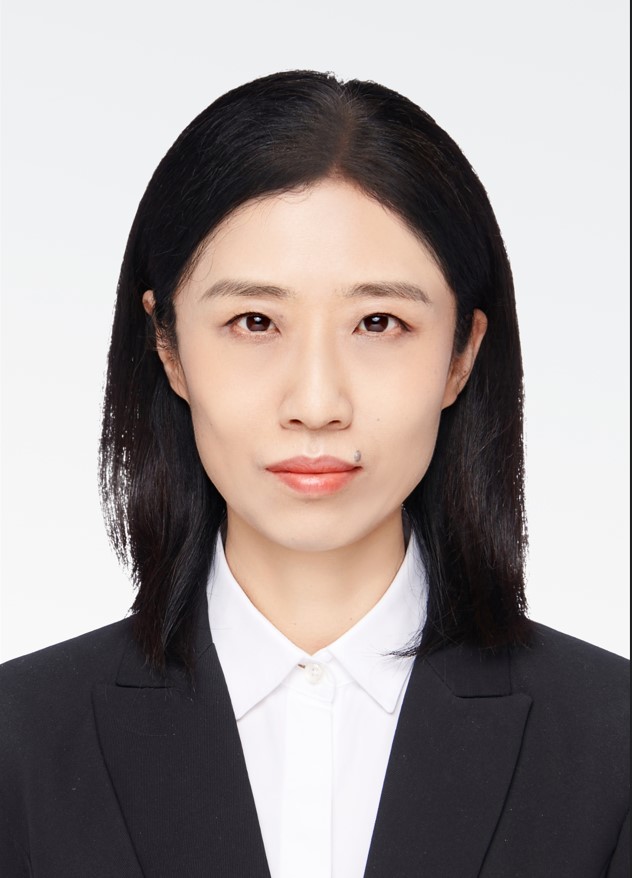}}]{Yan Li} received the Ph.D. degree in Management Science and Engineering from Tianjin University, Tianjin, China. She is an associate professor with Hebei University of Technology. She has authored or coauthored more than 20 research papers in some journals, such as IEEE Transactions on Knowledge and Data Engineering, IEEE Transactions on Cybernetics, ACM Transactions on Knowledge Discovery from Data, Information Sciences, and Applied Intelligence Her current research interests include data mining and supply chain management.
\end{IEEEbiography}

%\vspace{-1.5cm}

\begin{IEEEbiography}[{\includegraphics[width=1in,height=1.25in,clip,keepaspectratio]{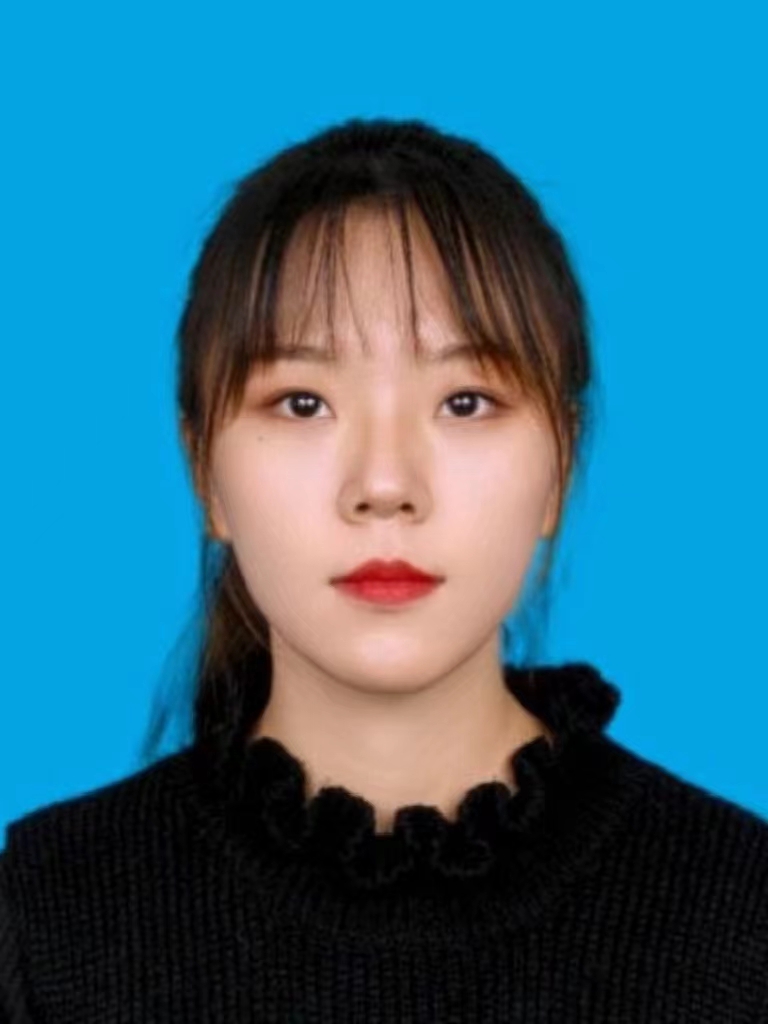}}]{Chenyu Ma} is a master degree candidate with the Management Science and Engineering, Hebei University of Technology. Her current research interests include data mining and machine learning.
%\vspace{-1cm}
\end{IEEEbiography}

\begin{IEEEbiography}[{\includegraphics[width=1in,height=1.25in,clip,keepaspectratio]{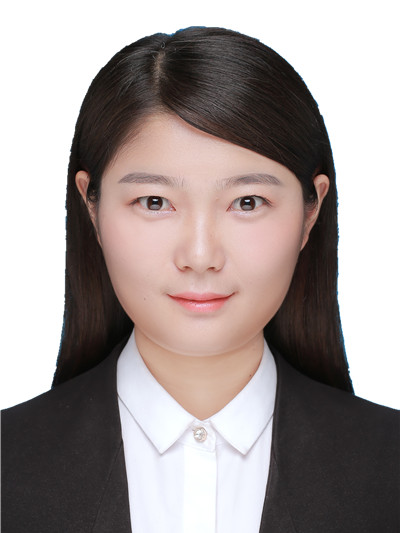}}]{Rong Gao}  received the Ph.D. degree in operational research from Tsinghua University, Beijing, China, in 2017. She is currently a Ph.D. Supervisor and an Associate Professor with the School of Economics and Management, Hebei University of Technology, Tianjin, China. She has authored or coauthored 60 articles on several journals including IEEE Transactions on Fuzzy Systems, Knowledge-Based Systems, Journal of Computational and Applied Mathematics, Fuzzy Optimization and Decision Making, etc. Her current research interests are in decision making, data analysis, and supply chain management under various kinds of uncertainty.

\end{IEEEbiography}

%\vspace{-1.5cm}
\begin{IEEEbiography}[{\includegraphics[width=1in,height=1.25in,clip,keepaspectratio]{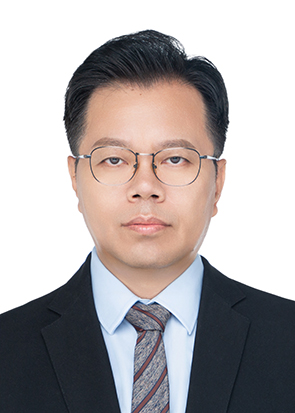}}]{Youxi Wu}received the Ph.D. degree in Theory and New Technology of Electrical Engineering from the Hebei University of Technology, Tianjin, China. He is currently a Ph.D. Supervisor and a Professor with the Hebei University of Technology. He has published more than 30 research papers in some journals, such as IEEE TKDE, IEEE TCYB, ACM TKDD, ACM TMIS, SCIS, INS, JCST, KBS, ESWA, JIS, Neurocomputing, and APIN. He is a distinguished member of CCF and a senior member of IEEE. His current research interests include data mining and machine learning.
\vspace{-1cm}
\end{IEEEbiography}

%\vspace{-1.5cm}
\begin{IEEEbiography}[{\includegraphics[width=1in,height=1.25in,clip,keepaspectratio]{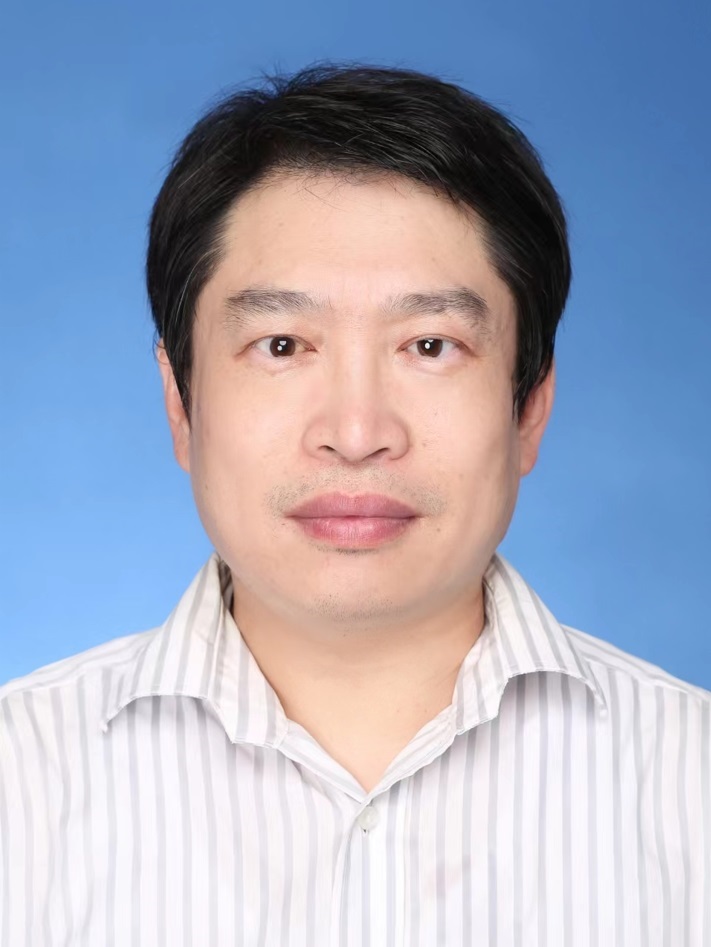}}]{Jinyan Li} is a Distinguished Professor at the Faculty of Computer Science and Control Engineering, Shenzhen Institute of Advanced Technology, Chinese Academy of Sciences. His research areas include bioinformatics, data mining and machine learning. He has published 260 papers at international conferences such as KDD, ICML, AAAI, PODS, ICDM, SDM, and ICDT, and at prestigious journals such as Artificial Intelligence, Machine Learning, TKDE, Bioinformatics, Nucleic Acids Research, Cancer Cell and Nature Communications. He is well known from his pioneering theories and algorithms on “emerging patterns”, a data mining concept for the discovery of contrast growing trends and patterns. Professor Li received a Bachelor degree of applied mathematics from National University of Defense Technology (1991), a Master degree of computer engineering from Hebei University of Technology (1994), and PhD of computer science and software engineering from the University of Melbourne (2001). 
%\vspace{-1cm}
\end{IEEEbiography}

%\vspace{-1.5cm}
\begin{IEEEbiography}[{\includegraphics[width=1in,height=1.25in,clip,keepaspectratio]{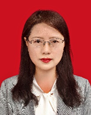}}]{Wenjian Wang}
received the BS degree in computer science from Shanxi University, China, in 1990, the MS degree in computer science from Hebei University of Technology, China, in 1993, and the Ph.D degree in applied mathematics from Xi'an Jiaotong University, China, in 2004. She is currently a full-time professor and a Ph.D supervisor with Shanxi University. She has published more than 260 academic papers. Her research interests include machine learning, data mining, computational intelligence, etc.

%\vspace{-1cm}
\end{IEEEbiography}

%\vspace{-1.5cm}
\begin{IEEEbiography}[{\includegraphics[width=1in,height=1.25in,clip,keepaspectratio]{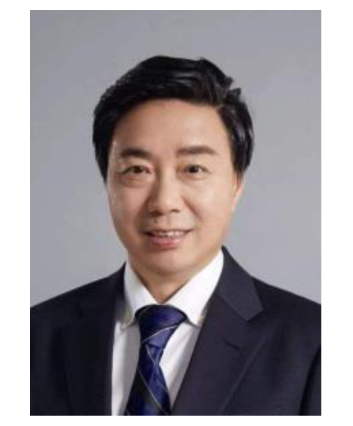}}]{Xindong Wu} (Fellow, IEEE) received the BS and MS degrees in computer science from the Hefei University of Technology, Hefei, China, in 1984 and 1987, respectively, and the PhD degree in artificial intelligence from The University of Edinburgh, Edinburgh, U.K., in 1993. He is currently a senior research scientist with the Research Center for Knowledge Engineering, Zhejiang Lab, China. He is a Foreign Member of the Russian Academy of Engineering, and a Fellow of AAAS. His research interests include data mining, knowledge engineering, Big Data analytics, and marketing intelligence.

\end{IEEEbiography}

\end{document}